\documentclass[a4paper,reqno]{amsart}

\usepackage{amssymb}
\usepackage{amsmath}
\usepackage{mathtools}
\usepackage{amsthm}
\usepackage[normalem]{ulem}
\usepackage{xcolor}
\usepackage{mathrsfs}
\usepackage{enumitem}
\usepackage{comment}
\usepackage[hidelinks]{hyperref}
\usepackage{cases}
\usepackage{cancel}

\newtheorem{theorem}{Theorem}[section]
\newtheorem{lemma}[theorem]{Lemma}
\newtheorem{proposition}[theorem]{Proposition}
\newtheorem{corollary}[theorem]{Corollary}
\theoremstyle{remark}
\newtheorem{remark}[theorem]{Remark}

\newcommand{\iu}{\mathrm{i}}
\newcommand{\E}{\mathbb{E}}

\newcommand{\Par}{\Theta}
\newcommand{\Pro}{\mathbb{P}}
\newcommand\dotw{\mathbin{\vcenter{\hbox{\scalebox{0.4}{$\bullet$}}}}}

\newcommand{\x}{{\bf x}}

\newcommand{\y}{{\bf y}}

\newcommand{\BCinf}[1]{BC^\infty(#1)}
\newcommand{\HE}{H_{b}^E}
 \newcommand{\HB}{H_{b}}
\newcommand{\HEb}{h_{b}^E}
\newcommand{\M}{\mathcal{M}}
\newcommand{\R}{\mathbb{R}}
\newcommand{\C}{\mathbb{C}}
\newcommand{\N}{\mathbb{N}}
\newcommand{\Z}{\mathbb{Z}}
\newcommand{\morph}[3]{#1\colon #2\to #3}
\newcommand{\Cinf}[1]{C^\infty(#1)}
\newcommand{\Ccinf}[1]{C_c^\infty(#1)}
\newcommand{\im}{\mathrm{i}}
\newcommand{\Hb}{(-\im\nabla -bA)^2}
\newcommand{\hb}{\mathfrak{H}_b^E}

\newcommand{\xbf}{\mathbf{x}}
\newcommand{\ybf}{\mathbf{y}}
\newcommand{\wbf}{\mathbf{w}}

\newcommand{\ebf}{\mathbf{e}}
\newcommand{\ip}[2]{\langle #1, #2 \rangle}
\newcommand{\dd}{\, \mathrm{d}}
\newcommand{\abs}[1]{\lvert #1\rvert}
\newcommand{\abslr}[1]{\left\lvert #1\right\rvert}
\newcommand{\norm}[1]{\lVert #1\rVert}
\newcommand{\jnorm}[1]{\langle #1\rangle}
\newcommand{\e}{\mathrm{e}}
\newcommand{\FT}{\mathscr{F}}

\newcommand{\s}{\mathbf{s}}

\newcommand{\tempdist}[1]{\mathscr{S}'(#1)}
\newcommand{\Schwartsfunct}[1]{\mathscr{S}(#1)}
\newcommand{\slowlyinc}[1]{\mathscr{O}_M(#1)}
\newcommand{\dist}{\mathrm{dist}}
\newcommand{\pair}[2]{\langle #1,#2\rangle_{d}}

\newcommand{\RbE}{(\HE+\lambda)^{-1}}
\newcommand{\Rb}{(\HEb+\lambda)^{-1}}

\newcommand{\RB}{R_{b}}
\newcommand{\RE}{R_{b}^E}
\newcommand{\HR}{\mathfrak{R}_{b}}
\newcommand{\HER}{\mathfrak{R}_{b}^E}

\newcommand\restr[2]{{
  \left.\kern-\nulldelimiterspace 
  #1 
  \vphantom{\big|} 
  \right|_{#2} 
  }}

\DeclareMathOperator{\Supp}{supp}

\colorlet{lightgray}{white!70!black!30!}

\begin{document}

\title[Regularity properties of bulk and edge current densities]{Regularity properties of bulk and edge \\current densities at positive temperature}


\author{Massimo Moscolari}

\author{Benjamin B. St{\o}ttrup}

\date{\today, ArXiv version.}

\begin{abstract}
We consider magnetic Schr\"odinger operators describing a quantum Hall effect setup both in the plane and in the half-plane. First, we study the structure and smoothness of the operator range of various powers of the half-plane resolvent. Second, we provide a complete analysis of the diamagnetic current density at positive temperature: we prove that bulk and edge current densities are smooth functions and we show that the edge current density converges to the bulk current density faster than any polynomial in the inverse distance from the boundary. Our proofs are based on gauge covariant magnetic perturbation theory and on a detailed analysis of the integral kernels of functions of magnetic Schr\"odinger operators on the half-plane.
\end{abstract}





\maketitle

\tableofcontents

\section{Introduction}

The quantum Hall effect is the archetypal example of a topological insulator. The peculiar feature of this exotic phase of matter is summarized in the bulk-edge correspondence, that is nowadays an ubiquitous mechanism that goes beyond condensed matter physics. In the context of the quantum Hall effect, the bulk-edge correspondence manifests itself as an exact equality between the bulk transverse conductivity and the edge conductance. At zero temperature, it has been shown that these objects are quantized and have a geometric origin \cite{SchulzBaldesKellendonkRichter,KellendonkRichterSchulzBaldes,KellendonkSchulzBaldes, ElgartGrafSchenker}.  Loosely speaking, the bulk-edge correspondence is an equality between quantum transport coefficients of bulk and edge Hamiltonians. While this correspondence has been proved at zero temperature both in discrete and continuous models, the first proof of bulk-edge correspondence at positive temperature appeared only recently in \cite{CorneanMoscolariTeufel}. Since there is a vast literature about bulk-edge correspondence and a complete discussion of the literature and the history of the subject is out of the scope of our paper, we refer the interested reader to the introduction of \cite{CorneanMoscolariTeufel} and to the monograph \cite{ProdanSchulzBaldes} for a detailed presentation.

The analysis of the bulk transverse conductivity is mainly based on transport theory in the linear response framework \cite{BES,AizenmanGraf, BoucletGerminetKleinSchenker}, which is a topic that has been extensively analyzed in the mathematical physics literature. On the other hand, the quantization of the edge conductance, which is defined as the statistical response with respect to the variation of the chemical potential, relies on a subtle balance of diamagnetic currents at the edge of the system. While there are several works regarding the analysis of the total edge current at zero temperature, see for example \cite{GrafFrohlichWalcher, deBievrePule, HislopSoccorsi}, there are much fewer mathematical works (see \cite{Kunz,MacrisMartinPule} and references therein) which have been concerned with a detailed description of diamagnetic current densities in the quantum Hall effect setting, and, most importantly, that holds true in the positive temperature setting. Our main purpose is to push forward this analysis by carefully analyzing the structure and the regularity properties of magnetic Schr\"odinger operators defined on the half-plane. In particular, this allows us to study the smoothness and the behavior at infinity of diamagnetic current densities at positive temperature (cf. Theorem~\ref{thm:smoothnes}~and~\ref{thm:currents} below), which play an important role in the bulk-edge correspondence analysis \cite{CorneanMoscolariTeufel}.

\subsection{The setting and the main results}

We consider a Hamiltonian describing a quantum Hall effect setup in the infinite volume limit, that is the bulk Hamiltonian. Let $V\in \BCinf{\R^2;\R}$ and  $\mathcal{A}\in \BCinf{\R^2;\R^2}$, i.e.\ $V$ and $\mathcal{A}$ are smooth and with uniformly bounded derivatives of all orders. Then, the bulk dynamics in $L^2(\mathbb{R}^2)$
is described by the Hamiltonian
\begin{equation*}
    H_{b}=\left(-\iu \nabla - \mathcal{A} - b A \right)^2+V
\end{equation*}
where $A$ is the magnetic potential of a constant unit magnetic field perpendicular to the plane,  which in the Landau gauge is given by $A(\x)\coloneqq\left(-x_2,0\right)$, $b\coloneqq-e \mathfrak{B}$ is a real parameter representing the strength of the constant magnetic field $\mathfrak{B}$ and $e$ is the elementary charge. It is a standard result \cite{ReedSimon2,FournaisHelffer} that under the previous hypothesis $H_b$ is essentially self-adjoint on $\Ccinf{\R^2}$. 

Then, the edge dynamics is described by a Hamiltonian defined on $L^2(E)$, where
\begin{equation*}
	E\coloneqq\left\{(x_1,x_2)\in \R^2 |\;  x_2 >  0 \right\}.
\end{equation*}
The edge Hamiltonian is denoted by
$H^E_{b}$ and equals the Dirichlet realization of $H_{b}$ in $E$. That is, $H^E_b$ is the Friedrichs extension of the quadratic form $Q^E_b$ defined by 
\begin{equation*}
    Q^E_{b}(f,g)=\ip{f}{H_b g}_{L^2(E)},
\end{equation*}
with form domain given by $C_c^{\infty}(E)$. However, to analyze the structure of the resolvent of $H_b^E$ we will take an alternative but fairly standard approach to defining the edge Hamiltonian. Specifically, if $\Ccinf{\bar{E}}$ denotes the space of smooth functions $f$ on $E$ with partial derivatives extending continuously to $\partial E$ and for which there exists $a>0$ such that $f(\xbf)=0$ for $\xbf\in \bar{E}\setminus( (-a,a)\times (0,a))$, then we define
\begin{equation}\label{eq:HbEdef}
    H^E_b\coloneqq\overline{\restr{H_b}{\Ccinf{\bar{E}}}},
\end{equation}
where $\restr{H_b}{\Ccinf{\bar{E}}}$ is regarded as an operator on $L^2(E)$. As shown in Theorem~\ref{thm:1}~\ref{thm:1Part1} below it turns out that the this definition coincides with the Friedrichs extension of $Q^E_b$ (see also the discussion in Remark~\ref{rem:1}).

\begin{theorem}\label{thm:1}
The following results hold true:
\begin{enumerate}[label=(\roman*)]
    \item \label{thm:1Part1} The operator $H_b$ restricted to $\Ccinf{\bar{E}}$ is essentially self-adjoint and its closure, namely $H_b^E$, coincides with the Friedrichs extension of the quadratic form $Q^E_{b}$. The domain $D(H^E_b)$ coincides with the range of the operator $S_b(-\lambda)$, which has an explicit integral kernel defined by \eqref{eq:Slambda1}.
    \item \label{thm:1Part2} Let $\Omega\subset \R$ be compact. Then for all $\lambda>0$ sufficiently large and $b\in \Omega$ there exists an integral operator $K_b(-\lambda)$, which is smooth with respect to $b\in \Omega$ in the operator norm topology, such that
    \begin{equation}\label{eq:HEexpansion}
        (H^E_b+\lambda)^{-1}(\xbf,\ybf)=\e^{\im b\phi(\xbf,\ybf)}K_b(-\lambda)(\xbf,\ybf),
    \end{equation}
    where $\phi$ denotes the Peierls phase given by $ \phi(\xbf,\ybf):=\frac{1}{2} (y_1-x_1)(x_2+y_2)$, $\x,\y \in E$.
\end{enumerate}
\end{theorem}
 \begin{remark}\label{rem:1}
    Notice that the main achievement of Theorem~\ref{thm:1} is not the construction of the operator itself, which may be obtained by other method, like the Feynman-Kac semigroup methods (as in \cite{MacrisMartinPule}) or by a partial Fourier transform which reduces the problem to a one-dimensional operator \cite[Chapter 4]{FournaisHelffer}. The main achievement of Theorem~\ref{thm:1} is formula \eqref{eq:HEexpansion}. This will allow us to investigate the regularity properties of various physical edge quantities. Furthermore, since our method is based on gauge covariant magnetic perturbation theory \cite{CN1998,N2002,CN,Cornean2010,CP2012}, it may easily be extended to non-constant but globally bounded magnetic fields.
 \end{remark}
 \begin{remark}
 The phase factor appearing in \eqref{eq:HEexpansion} is the only obstruction for $(\HE+\lambda)^{-1}$ to be smooth with respect to $b$ in the operator norm topology. However, $(\HE+\lambda)^{-1}$ is smooth with respect to $b$ in the strong operator topology (cf. \cite[Corollary 1.2]{CorneanMonacoMoscolari})
 \end{remark}

By using the magnetic pseudodifferential calculus \cite{MP,IMP07,IMP10,CHP2018,CGSS}, it is possible to prove smoothing properties for the resolvent of the bulk Hamiltonian, i.e.\ for $\lambda>0$ sufficiently large and any $m\in \N$, there exists $N$ large enough such that $(H_{b}+\lambda)^{-N}$ maps $L^2(\R^2)$ continuously to $C^m(\R^2)$. However, such techniques do not apply to operators defined on a domain with boundaries, like in the case of the edge Hamiltonian $H^E_{b}$. As a first key technical result, we show a smoothing property for the resolvent of the edge Hamiltonian.

\begin{proposition}\label{prop:2}
Let $g\in \Ccinf{E}$ and let $X_j$, $j\in \{1,2\}$, be the standard position operators in $L^2(E)$. There exists $\lambda>0$ sufficiently large such that for every $m\in \N$, there exist $N, \epsilon>0$ with the property that the operators $g (\HE+\lambda)^{-N} \e^{\epsilon|X|}$ and $g\; \im [\HE,X_j](\HE+\lambda)^{-N} \e^{\epsilon|X|}$  map $L^2(E)$ continuously into $C^m(E)$. Specifically, there exists $C>0$ such that for every $k\in \{0,1\}$ and $f\in L^2(E)$:
\begin{equation*}
    \norm{g (\im [\HE,X_j])^k\left(\HE+\lambda\right)^{-N} \e^{\epsilon|X|}f}_{C^m(E)}\leq C\norm{f}_{L^2(E)}.
\end{equation*}
\end{proposition}

As it has been shown in \cite{CorneanMoscolariTeufel}, when $F\colon \sigma(\HE)\to \C$ is a function which has an extension in $\Schwartsfunct{\R}$, the edge operator $\iu [H^{E}_{b },X_1]  F (H^{E}_{b })$  has a jointly continuous integral kernel. As a consequence of Proposition \ref{prop:2} we have the following optimal result:

\begin{theorem}
\label{thm:smoothnes}
Let $F\colon \sigma(\HE)\to \C$ be a function which has an extension in $\Schwartsfunct{\R}$.
For all $g_1,g_2\in C_c^\infty (E)$ and $j\in \{1,2\}$, the operators 
\begin{equation*}
	g_1F(\HE)g_2 \, , \qquad g_1\left(\iu\left[\HE,X_j\right]F(\HE)\right)g_2
\end{equation*}
have smooth integral kernels.
\end{theorem}

\begin{remark}\label{rem:HB}
    Even though we only deal with the edge operator $H_{b}^E$, the proofs of Proposition~\ref{prop:2} and Theorem~\ref{thm:smoothnes} and related results in Section~\ref{sec:regularity} and \ref{sec:GPT} hold (and are considerably simpler) for the bulk operator as well. 
\end{remark}

Now let $F_\mathrm{FD}$ denote the Fermi-Dirac distribution, i.e., $F_\mathrm{FD}(x)=(\e^{(x-\mu)/T}+1)^{-1}$, where $\mu\in \R$ is the Fermi energy and $T>0$ is the temperature in suitable units. Furthermore, let $F$ denote a Schwartz function such that $F(x)=F_{\mathrm{FD}}(x)$ for all $x\in \left[-1+\inf \left(\sigma(\HB)\right),+\infty\right)$. 
We consider the bulk and edge (particle-)current density $\mathcal{J}_{1,b}$ and $\mathcal{J}_{1,b}^{E}$ defined on $\R^2\times \R^2$ and $E\times E$, respectively, by
\[\mathcal{J}_{1,b}^{\;\;/E}(x_1,x_2):=\left(\iu [H^{\;\;/E}_{b},X_1]  F(H^{\;\;/E}_{ b })\right)\left((x_1,x_2),(x_1,x_2)\right).\]
Here $H^{\;\;/E}_{b }$ is a shorthand notation for either $\HB$ or $\HE$. Note that by Theorem~\ref{thm:smoothnes} these two functions are well defined.
\begin{theorem}
\label{thm:currents}
$\mathcal{J}_{1,b}^{\; \;/E}: E \to \C$ are smooth  functions such that
\begin{equation*}
    \sup_{x_1\in \R} \mathcal{J}_{1,b}^E(x_1,x_2) -\mathcal{J}_{1,b}(x_1,x_2) = \mathcal{O}(x_2^{-\infty}) \text{ for } x_2\to +\infty.
\end{equation*}
\end{theorem}

{ Notice that Theorem~\ref{thm:currents} holds true if $F$ is any function as in Theorem~\ref{thm:smoothnes}. The use of $F_{\mathrm{FD}}$ above is for $\mathcal{J}_{1,b}^{\;\;/E}$ to have a physical interpretation \cite{CorneanMoscolariTeufel}.}

Theorem~\ref{thm:currents} shows that the edge particle-current density converges to the bulk particle-current density faster than any polynomial in the inverse distance from the boundary. This characterization of the edge current density is, on the one hand important to gain an heuristic picture of the bulk-edge correspondence at positive temperature and, on the other hand, plays a crucial role in the developing of the proper definition of the total edge current \cite[Section 1.3]{CorneanMoscolariTeufel}.

Let us briefly recall the setting of \cite{CorneanMoscolariTeufel}. We assume that $\mathcal{A}$ is $\Z^2$-periodic, and we assume that $V$ is made of two contributions, that is $V=V_{\Z^2} + V_{\omega}$, where $V_{\Z^2}$ is $\Z^2$-periodic and $V_{\omega}$ is a random potential defined as follows. Let ${\mathrm p}$ be a probability measure on $[-1,1]$ and define $\Par=[-1,1]^{\Z^2}$. Consider $\Pro=\bigotimes_{\Z^2} {\mathrm p}$ to be a probability measure on $\Par$. We denote by $\E$ the expectation on the probability space $(\Par,\Pro)$. For every $\omega \in \Par$ we define a random potential $V_{\omega}$, by 
\begin{equation*}
    V_{\omega}(x)=\sum_{\gamma \in \Z^2} \omega_\gamma u(x-\gamma)
\end{equation*}
where $u \in\Ccinf{\R^2}$. By construction, $\{\omega_i\}_{i \in \Z^2}$ is a family of independent identically distributed (with probability distribution ${\mathrm p}$) random variables taking values in $[-1,1]$.

Then the bulk, respectively edge, Hamiltonian with the potentials satisfying such assumptions is denoted by $H_{\omega,b}$, resp. $H^E_{\omega,b}$. Likewise the corresponding particle-current densities are denoted by $\mathcal{J}_{1,\omega,b}$ and $\mathcal{J}_{1,\omega,b}^E$. It is straightforward to check that the family of self-adjoint operators $\{H_{\omega,b}\}_{\omega \in \Par}$ is ergodic with respect to the lattice $\Z^2$, while the family $(H^E_{\omega,b})_{\omega \in \Par}$ is ergodic with respect to the one-dimensional lattice generated by the vector $(1,0)$.

In view of the ergodicity of the potential we can define the averaged bulk and edge particle-current density as
\begin{equation*}
    j^{\;\;/E}_{1,b}(x_2):= \int_{0}^{1} \mathbb{E}[\mathcal{J}_{1,\dotw, b}^{\;\;/E}](x_1,x_2) \mathrm{d} x_1 .
\end{equation*}
Then, Theorem~\ref{thm:currents} implies 
\begin{equation*}
    j_{1,b}^E(x_2) -j_{1,b}(x_2) = \mathcal{O}(x_2^{-\infty}) \text{ for } x_2\to +\infty .
\end{equation*}

\begin{remark}
\label{rmk:PureLandau}
It is interesting to consider the case of the pure Landau Hamiltonian, that is $\mathcal{A}\equiv0$ and $V\equiv0$. In that case, as a consequence of magnetic translation invariance we get that $j_{1,b}$ must be constant, however the vanishing of the persistent current implies that $j_{1,b}$ is equal to zero. Thus, $j_{1,b}$ is identically zero and Theorem~\ref{thm:currents} shows that the edge current converges to zero in the bulk faster than any  polynomial in the inverse distance from the edge.
\end{remark}

\begin{remark} 
The result of Theorem~\ref{thm:currents} should be compared with the results in \cite{Kunz,MacrisMartinPule}. In both papers the authors analyze the behavior of diamagnetic currents far from the boundary and they also connect the total edge current with the bulk magnetization (see \cite{CorneanMoscolariTeufel} for the connection with the bulk magnetization in our setting). They consider only the pure Landau operator and $F$ to be the Maxwell-Boltzmann distribution. These two conditions are crucial for them since their approach relies on the Feynmann-Kac formula applied to the Gibbs semigroup. In such particular setting (c.f. Remark \ref{rmk:PureLandau}), they obtain a Gaussian decay for the current. Moreover, notice that they also extend their result to the interacting setting in a perturbative way. Therefore, if on the one hand we significantly generalize the results in \cite{Kunz,MacrisMartinPule} by considering more general class of operators and functions, on the other hand we have that the decay of Theorem~\ref{thm:currents} is not optimal. An analysis of the optimal decay of the diamagnetic current is postponed to a future work.  
\end{remark}

\subsection{Open questions}

\begin{enumerate}[label=(\roman*)]
    \item A first open problem concerns the extension of our results to infinite domains with a boundary given by the graph of a smooth function, see for example \cite{GrafFrohlichWalcher}. The general strategy developed in our paper should also apply in that setting, provided one performs a detailed preliminary analysis of the integral kernel of the Dirichlet Laplacian on such domains. After that, it would be interesting to extend our results to other boundary conditions, like Neumann boundary conditions \cite{FournaisHelffer}. In such cases, the situation is even more delicate and one has to suitably modify the gauge covariant magnetic perturbation theory. We plan to address these issues in a future work.   
    An even more challenging problem would be to extend our results to general elliptic operators, including Dirac operators.
    
    \item The bulk-edge correspondence has been extensively analyzed in the one-particle approximation. One of the main open problems in mathematical physics regards the understanding of bulk-edge correspondence for interacting systems. While perturbative results in the discrete setting are already present in the literature \cite{AntinucciMastropietroPorta,MastropietroPorta}, no general proof exists. Therefore, extending our results to many-body magnetic Schr\"odinger operators, in the spirit of \cite{Kunz,MacrisMartinPule}, would be a significant step forward.
    
\end{enumerate}

\addtocontents{toc}{\protect\setcounter{tocdepth}{0}}

\subsection*{Content of the paper}
Section~\ref{sec:selfadjoint} is devoted to the proof of Theorem~\ref{thm:1}, which is divided in two main parts. In the first part, that is the proof of Theorem~\ref{thm:1}\ref{thm:1Part1}, we show that $H^E_b$ is self-adjoint, we analyze the Friedrichs extension of the quadratic form $Q^E_b$ and show that it coincides with $H^E_b$. Then, in the second part we prove Theorem~\ref{thm:1}\ref{thm:1Part2} by carefully analyzing the structure of the integral kernel of the resolvent of the edge Hamiltonian. In Appendix~\ref{s:dist} we collect some useful standard results regarding the formal manipulation of distributions that we use in the proof of Theorem~\ref{thm:1}. 

In Section~\ref{sec:regularity} we prove Proposition~\ref{prop:2} and Theorem~\ref{thm:smoothnes}. The proofs are based on a detailed analysis of the integral kernels of functions of the edge and bulk operators. In order to ease the reading, we postpone this technical analysis to Appendix~\ref{s:intop}.

Finally, in Section~\ref{sec:GPT} we prove Theorem~\ref{thm:currents} by exploiting geometric perturbation theory together with the integral kernel estimates in Appendix~\ref{s:intop}.

\subsection*{Acknowledgments}
The authors would like to thank H. D. Cornean and S. Teufel for many fruitful discussions and helpful suggestions. The work of MM is supported by a fellowship of the Alexander von Humboldt Foundation. BBS is especially grateful to S. Teufel for extending an invitation to visit the University of Tübingen during the fall of 2021 and to ``Anders Brøndums Fond'' for supporting the visit financially.

\subsection*{Notation}
\label{sec:notation}
For the convenience of the reader we list here the main notations used in the paper:
\begin{itemize}
    \item Euclidean norms are denoted by $\abs{\cdot}$. Other norms (such as $L^p$ norms, operator norms, etc.) are denoted by $\norm{\cdot}$. When ambiguity arise we denote norms with a subscript indicating the associated space.
    \item The Euclidean inner product is denoted by $\xbf\cdot \ybf$ and we use the shorthand $\xbf^2$ for $\xbf\cdot \xbf$.
    \item For any open set $\Omega\subset \R^2$ the inner product on $L^2(\Omega)$ is denoted by $\ip{\cdot}{\cdot}_{L^2(\Omega)}$, where we omit the subscript when it does not lead to any confusion. We take the inner product to be linear in the second variable.
    \item The dual pairing of $\tempdist{\R ^2}$ and $\Schwartsfunct{\R^2}$ is denoted by $\pair{\cdot}{\cdot}$ and it is taken to be linear in both variables.
    \item The bracket $\jnorm{\cdot}$ is defined by $\jnorm{\xbf}\coloneqq(1+\abs{\xbf}^2)^{1/2}$, for $\xbf\in \R^n$.
    \item For differentiable functions on $\R^2$ we use the notation $\partial_j\coloneqq \frac{\partial}{\partial x_j}$, $j\in \{1,2\}$.
    \item For functions defined on $\Omega\times \Omega$, where $\Omega$ is an open set $\Omega\subset \R^2$, we use $\partial_1$ and $\partial_2$ to denote partial derivatives with respect to the first and second variable, respectively, i.e.\
    \begin{equation*}
        \partial_1^\alpha \partial_2^\beta f(\xbf,\ybf)=\frac{\partial^\alpha}{\partial\xbf^\alpha} \frac{\partial^\beta}{\partial\ybf ^\beta}f(\xbf,\ybf),
    \end{equation*}
    for any multi-indices $\alpha,\beta\in \N_0^2$. Similar notation is used for coordinate-wise gradients and Laplacians. This should not be confused with the notation of the previous item in this list.
    \item For functions of more than two variables, we may use the variable in question as a subscript for the differential operator, e.g.\ $\partial_\xbf$.
\end{itemize} 

\addtocontents{toc}{\protect\setcounter{tocdepth}{2}}

\section{Proof of \texorpdfstring{Theorem~\ref{thm:1}}{Theorem 1.1}}
\label{sec:selfadjoint}
\subsection{Proof of \texorpdfstring{Theorem~\ref{thm:1}\ref{thm:1Part1}}{Theorem 1.1(i)}}
The aim of this section is to provide a detailed analysis of the edge Hamiltonian $H^E_b$ in order to prove Theorem~\ref{thm:1}\ref{thm:1Part1}. We start by showing that $\Hb$ defined on $\Ccinf{\bar{E}}$ is essentially self-adjoint and that $H^E_b$ is a relatively bounded perturbation of $\overline{\Hb}$. This result is useful to obtain a norm convergent resolvent expansion (see Theorem \ref{thm:1} \ref{thm:1Part2}) which is interesting on its own and will be used in the rest of the paper. All operators in this section are considered as operators on $L^2(E)$ unless otherwise stated. Likewise, all integral kernels are defined on $E\times E$ unless otherwise stated.

Let us first precisely define the domains of interest. Let $a>0$ and define the set $U_a\coloneqq (-a,a)\times (0,a) \subset E$. We define the space of functions $C^{\infty}(\bar{E}) \subset C^{\infty}(E)$ by
\begin{equation*}
    \Cinf{\bar{E}}=\{ f \in C^{\infty}(E)\mid  \partial^\alpha f \textup{ extends continuously to } \bar{E} \textup{ for all }\alpha\in \N_0^2\},
\end{equation*}
and recall that $\Ccinf{\bar{E}}$ is defined as
\begin{equation*}
    \Ccinf{\bar{E}}=\{f\in \Cinf{\bar{E}} \mid \exists \, a>0\textup{ s.t.  } f(\xbf)=0\textup{ for all } \xbf\in \bar{E}\setminus U_a\}.
\end{equation*}
Let $\hb$ denote the operator $\Hb$ defined on the dense domain $D(\hb)=\Ccinf{\bar{E}}$, where $b\in \R$ and $A(\xbf)=(-x_2,0)$. Furthermore, for $\mathcal{A}\in \BCinf{\R^2;\R^2}$ and $j\in \{1,2\}$, let $P_j$ denote the operator
\begin{equation*}
    P_j=-\im \partial_j-bA_j-\mathcal{A}_j
\end{equation*}
defined on the dense domain $D(P_j)=\Ccinf{\bar{E}}$. For $P=[P_1,P_2]$ and a scalar potential  $V\in\BCinf{\R^2;\R}$, define
\begin{equation}\label{eq:W}
    W=-2\mathcal{A}\cdot P-\im(\nabla\cdot\mathcal{A})+\mathcal{A}^2+V
\end{equation}
with domain $D(W)=\Ccinf{\bar{E}}$ and note that  

\begin{equation*}
    \hb+W=\restr{H_b}{\Ccinf{\bar{E}}}=\restr{H^E_b}{\Ccinf{\bar{E}}}.
\end{equation*}
\begin{lemma}\label{lem:sym}
    The operators $\hb$, $P_j$, and $W$ are symmetric. Additionally, $\hb$ is a positive operator.
\end{lemma}
\begin{proof}
The result follows by explicit integration together with the use of Green's formula.
\end{proof}

\subsubsection{Self-adjointness of \texorpdfstring{$\HE$}{HbE}}
\label{subsec:saHE}
Let us denote by $h_b^E$ the operator $h_b^E:=\overline{\hb}$. The proof of the self-adjointness of $H_b^E$ requires a standard application of the Kato-Rellich theorem (cf.\ \cite[Theorem X.12]{ReedSimon2}). As preliminary necessary results we show that $h_b^E$ is self-adjoint and that $\overline{W}$ is $h_b^E$-bounded.

Let us start with a technical lemma characterizing the domain of $h_b^E$.

\begin{lemma}\label{lem:MinDHb}
    Let $\M$ be the set of all functions $f\in \Cinf{\bar{E}}$ such that $f(x_1,0) \equiv 0$ and which decay exponentially together with all their derivatives, i.e. there exists $c>0$ 
    such that for all $\alpha\in \N_0^2$ and $\xbf\in \bar{E}$,
\begin{equation}\label{eq:Mdef}
    \abs{\partial^\alpha f(\xbf)}\leq C \e^{-c \abs{\xbf}}
\end{equation}
for some constant $C$. Then \(\M\subseteq D(h^E_b)\).
\end{lemma}
\begin{proof}
    Let $f\in \M$ be arbitrary. Let $\chi\colon \R \to [0,1]$ be a smooth function such that $\chi(x)=1$ for $\abs{x}\leq 1$ and $\chi(x)=0$ for $\abs{x}\geq 2$. For any $n\in \mathbb{N}$ define $\chi_n(\xbf)=\chi(x_1/n)\chi(x_2/n)$, and $f_n=\chi_n f$. By the properties of $f$ and $\chi$, Lebesgue's dominated convergence theorem gives that 
          $  \lim_{n\to \infty} \norm{\partial^\alpha f_n-\partial^\alpha f}_{L^2(E)}=0,$
        for all $\alpha\in \N_0^2$. In particular,
    \begin{equation}
    \label{eq:explicitHb}
        \hb f_n=-\Delta f_n+2ibA\cdot \nabla f_n+b^2A^2 f_n\to -\Delta f+2ibA\cdot \nabla f+b^2A^2 f 
    \end{equation}
    in $L^2(E)$ as $n\to \infty$, which shows that $f$ is in the closure of $\hb$.
\end{proof}

In order to show that $h_b^E$ is self-adjoint we have to analyse its range, we do this by exploiting gauge covariant magnetic perturbation theory (cf.\ \cite{CN1998,N2002,Cornean2010,CP2012}) to construct an almost resolvent for the operator $h_b^E$.

Let us start by recalling some known facts about the integral kernel of the standard Laplacian, that is $(-\Delta+\lambda)^{-1}$ (defined as an operator on $L^2(\R^2)$). For any $\lambda>0$, an application of the Fourier transform shows that $(-\Delta+\lambda)^{-1}$ has an integral kernel
\begin{equation*}
    (-\Delta+\lambda)^{-1}(\xbf,\ybf)=\FT^{-1}\bigg( \frac{1}{\abs{\xi}^2+\lambda} \bigg)(\xbf-\ybf)=\frac{1}{2\pi}K_0(\sqrt{\lambda}\abs{\xbf-\ybf}),
\end{equation*}
where $K_0$ is the MacDonald function (see Appendix~\ref{s:M}). By using the reflection method, the integral kernel of the resolvent of the free Laplacian defined with Dirichlet boundary condition on $E$ is given by 
\begin{equation*}
    R_E(-\lambda)(\xbf,\ybf)\coloneqq(-\Delta+\lambda)^{-1}(\xbf,\ybf)-(-\Delta+\lambda)^{-1}(\xbf,\ybf^*),
\end{equation*}
where $\ybf^*\coloneqq (y_1,-y_2)$. Consider now the Peierls phase $\phi$ defined as
\begin{equation}\label{eq:phi}
    \phi(\xbf,\ybf)=-\int_\xbf^\ybf A(\s)\cdot\dd \s=-\frac{1}{2} (x_1-y_1)(x_2+y_2). 
\end{equation}
Let $S_b(-\lambda)$ and $T_b(-\lambda)$ be the operators on $L^2(E)$ defined by the following integral kernels
\begin{equation}\label{eq:Slambda1}
    S_b(-\lambda)(\xbf,\ybf)\coloneqq\e^{\im b \phi(\xbf,\ybf)}R_E(\xbf,\ybf),
\end{equation}
and 
\begin{equation}\label{eq:Tlambda}
    T_b(-\lambda)(\xbf,\ybf)\coloneqq \e^{\im b\phi(\xbf,\ybf)}bA_t(\xbf-\ybf)\cdot
    [2 \im\nabla_{1}+bA_t(\xbf-\ybf)] R_E(\xbf,\ybf)
\end{equation}
where $A_t(\xbf)\coloneqq \frac{1}{2}(-x_2,x_1)$ denotes the transverse gauge. These two operators play a central role in our analysis and we study their properties in Appendix~\ref{s:M}.  In particular, we have that $S_b(-\lambda)$ and $T_b(-\lambda)$ are bounded operators (see Lemma~\ref{lem:S} and Lemma~\ref{lem:T}) and that $S_b(-\lambda)$ maps $C^{\infty}_c(E)$ into $\M$ (see Lemma~\ref{lem:SM}). We are now ready to state the following central technical proposition.

\begin{proposition}\label{prop:1}
    For any $f\in L^2(E)$, 
        \begin{equation}\label{eq:HbSlambda}
            \left(h^E_b+\lambda\right)S_b(-\lambda)f=f+T_b(-\lambda)f.
        \end{equation}
\end{proposition}

\begin{proof}
    Note first that since $h^E_b+\lambda$ is closed and both $S_b(-\lambda)$ and $T_b(-\lambda)$ are bounded (cf.\ Lemma~\ref{lem:S} and Lemma~\ref{lem:T}), it suffices to show \eqref{eq:HbSlambda} for $f\in \Ccinf{E}$. Hence, let $f\in \Ccinf{E}$ be arbitrary.
    
    By Lemma~\ref{lem:MinDHb} and the fact that $S_b(-\lambda)$ maps $C^{\infty}_c(E)$ into $\M$ (cf.\ Lemma~\ref{lem:SM}), it follows that $(h_b^E+\lambda)S_b(-\lambda)f$ is a continuous function, point-wise given by
    \begin{equation*}
        \left(\left(h_b^E+\lambda \right)S_b(-\lambda)f\right)(\xbf)=[(-\Delta+\lambda)+2\im bA(\xbf)\cdot \nabla+b^2A(\xbf)^2 ]\left(S_b(-\lambda)f\right)(\xbf),
    \end{equation*}
    for $\xbf\in E$.
    If $S_b(-\lambda)$ had a kernel in e.g.\ $\BCinf{E\times E}$, then the above derivatives could easily be calculated by using Lebesgue's dominated convergence theorem. However, since the derivatives of $S_b(-\lambda)$ have singularities on the diagonal $\xbf=\ybf$ (cf.\ the properties of $K_0$ given in \eqref{eq:K1}-\eqref{eq:K4}), we  compute the derivatives in the distributional sense, as in Appendix~\ref{s:dist}. 
    Thus, in the following we consider $f$ to be a function on $\R^2$ by extending it by zero and let $f^*$, and $\phi^*$ be the functions defined by 
    \begin{equation*}
        f^*(\xbf)=f(\xbf^*),\quad \phi^*(\xbf,\ybf)=\phi(\xbf,\ybf^*).
    \end{equation*}
    At this point it is useful to write
\begin{equation*}
    \left(S_b(-\lambda)f\right)(\xbf)=\pair{\mathfrak{f}_{1,0}}{e^{\im b\phi(\xbf,\xbf-\cdot)}f(\xbf-\cdot)}-\pair{\mathfrak{f}_{1,0}}{e^{\im b\phi^*(\xbf,\xbf-\cdot)}f^*(\xbf-\cdot)},
\end{equation*}
    where $\mathfrak{f}_{1,0}$ is the distribution corresponding to $(2\pi)^{-1}K_0(\sqrt{\lambda}\abs{\cdot})$ (cf.\ the notation introduced in Appendix~\ref{s:kernel}).
    Then, by applying the generalized Leibniz rule in Lemma~\ref{lem:diff2} together with the identity 
\begin{equation*}
    A(\xbf)=A_t(\xbf-\ybf)+\nabla_{1} \phi(\xbf,\ybf),
\end{equation*} 
and the distributional equality $(-\Delta+\lambda)\mathfrak{f}_{1,0}=\delta$, we get the explicit expression
    \begin{align*}
        &\left(\left(h_b^E+\lambda \right)S_b(-\lambda)f\right)(\xbf)\\
        &= f(\xbf)+\pair{2\im b\e^{\im b \phi(\xbf,\cdot)}A_t(\xbf-\cdot)(\nabla \mathfrak{f}_{1,0})_{\xbf}+b^2\e^{\im b \phi(\xbf,\cdot)}A_t(\xbf-\cdot)^2(\mathfrak{f}_{1,0})_{\xbf} }{f}\\
        &\quad+\pair{2\im b\e^{\im b \phi^*(\xbf,\cdot)}A_t(\xbf-(\cdot)^*)(\nabla \mathfrak{f}_{1,0})_{\xbf}+b^2\e^{\im b \phi^*(\xbf,\cdot)}A_t(\xbf-(\cdot)^*)^2(\mathfrak{f}_{1,0})_{\xbf} }{f^*},
    \end{align*}
    where we have used the fact that $f^*(\xbf)=0$ for $\xbf\in E$ and the notation $\Lambda_\xbf$ for $\Lambda\in \tempdist{\R^2}$ defined in \eqref{eq:Lambdax}.
    For $\abs{\alpha}=1$, the properties of $K_0$ in \eqref{eq:K1}--\eqref{eq:K4} show that $\partial^\alpha K_0(\sqrt{\lambda}\abs{\xbf})$ is integrable. Hence, it follows from a standard approximation argument that $\nabla \mathfrak{f}_{1,0}$ equals $(2\pi)^{-1}\nabla K_0(\sqrt{\lambda}\abs{\cdot})$ and a suitable change of variables gives \eqref{eq:HbSlambda}.
\end{proof}
\medskip
As an immediate consequence of Proposition~\ref{prop:1} we get the following corollary. 
\begin{corollary}\label{cor:esssa}
    The operator $h_b^E$ is essentially self-adjoint. 
\end{corollary}
\begin{proof}
From Lemma~\ref{lem:T} we have that $1+T_b(-\lambda)$ is invertible when $\lambda$ is sufficiently large. For such $\lambda$, Proposition~\ref{prop:1} implies that $\HEb+\lambda$ has a bounded inverse given by
\begin{equation}\label{eq:HE_resolvent}
    (\HEb+\lambda)^{-1}=S_b(-\lambda)(1+T_b(-\lambda))^{-1}=\sum_{n=0}^\infty S_b(-\lambda)(-T_b(-\lambda))^n.
\end{equation}
Since $h_b^E$ is symmetric and positive we conclude by applying a standard self-adjointness criterion, see for example \cite[Theorem 9.24]{Hall}.
\end{proof}
By Lemma~\ref{lem:sym} and Corollary~\ref{cor:esssa}, $(h_b^E+\lambda)^{-1}$ exists for all $\lambda>0$ and the $h_b^E$-boundedness of $\overline{P_1}$, $\overline{P_2}$, and $\overline{W}$ can easily be established.
\begin{lemma}\label{lem:W}
    Let $Y\in \{\overline{P_1},\overline{P_2},\overline{W}\}$, then $D(h_b^E)\subset D(Y)$ and there exists a constant $C$ such that  
    \begin{equation}\label{eq:W-Hb-bounded}
        \norm{Y(h_b^E+\lambda)^{-1}}\leq C\frac{\sqrt{1+\lambda}}{\lambda}
    \end{equation}
    for any $\lambda>0$. In particular,  $\overline{P_1}$, $\overline{P_2}$, and $\overline{W}$ are infinitesimally $h_b^E$-bounded.
\end{lemma} 
\begin{proof}
    We only do the proof for $Y=\overline{W}$ since the other two cases are similar. For any $f\in D(W)=D(\hb)$ we have that 
    \begin{equation}\label{eq:W-estimate}
        \norm{Wf}^2\leq C\int_E \abs{(-\im\nabla -bA(\xbf))f(\xbf)}^2 \dd \xbf +C\norm{f}^2=C\ip{f}{\hb f}+C\norm{f}^2
    \end{equation}
    for some constant $C$. By a standard limiting argument it follows that \eqref{eq:W-estimate} holds true when $\hb$ and $W$ are replaced by $h_b^E$ and $\overline{W}$, respectively. Then, \eqref{eq:W-Hb-bounded} follows from \eqref{eq:W-estimate} by taking $\psi\in L^2(\R^2)$ and writing $f=(h_b^E+\lambda)^{-1}\psi$.
    
    By choosing $\lambda$ large enough in \eqref{eq:W-Hb-bounded}, the proof is concluded.
\end{proof}
By combining the results of Corollary~\ref{cor:esssa} and of Lemma~\ref{lem:W} we can use the Kato-Rellich theorem to conclude that 
\begin{equation}\label{eq:hboE}
    \overline{\hb+W}=H^E_b= h^E_b+\overline{W}
\end{equation}
is a self-adjoint operator on $D(h^E_b)$.

In order to conclude the proof of Theorem~\ref{thm:1}\ref{thm:1Part1} it remains to show that $H^E_b$ coincides with the Friedrichs extension of the quadratic form $Q_b^E$.

\subsubsection{ The Friedrichs extension of \texorpdfstring{$Q^E_b$}{QbE}}
\label{subsec:Friedrichs}
Recall that $Q^E_b$ is the quadratic form, with form domain $\Ccinf{E}$ (note that the form domain is not $\Ccinf{\bar{E}}$), given by 
\begin{equation*}
    Q_{b}^E(f,g)=\ip{f}{H_b \, g}_{L^2(E)} =\ip{f}{\HE \, g}_{L^2(E)}.
\end{equation*}
It is not difficult to see that $Q_{b}^E$ is semibounded, i.e.\ $Q_{b}^E(f,f)\geq -M\norm{f}^2$ for some $M$. Let $\mathcal{U}$ denote the completion of $\Ccinf{E}$ with respect to the inner product 
\begin{equation*}
    \ip{f}{g}_{Q^E_b}:=Q_{b}^E(f,g)+(1+M)\ip{f}{g}_{L^2(E)},
\end{equation*}
and let $q_{b}^E$ denote the extension of this inner product to all of $\mathcal{U}$. 
 Recall that the Friedrichs extension $H_{b,F}^E$ of $Q_b^E$ is defined by
    \begin{align*}
        &D(H_{b,F}^E)=\{f\in \mathcal{U}\mid \exists \, h\in L^2(E) \textup{ s.t. } q_{b}^E(g,f)=\ip{g}{h}_{L^2(E)} \textup{ for all } g\in \mathcal{U}\},\\
        &H_{b,F}^Ef=h-(1+M)f, \qquad f \in D(H_{b,F}^E) .
    \end{align*}
By \eqref{eq:hboE} and the self-adjointness of $H_{b,F}^E$, in order to show that $H^E_{b}$ coincides with $H_{b,F}^E$, it is sufficient to show that $H_{b,F}^E$ is an extension of $\hb+W$. For this purpose, we need the following technical lemma.
\begin{lemma}\label{lem:ipHb}
    Let $\chi\colon [0,\infty)\to [0,1]$ be a smooth function such that $\chi(x)=0$ for $x\leq1$ and $\chi(x)=1$ for $x\geq 2$. For any $f\in \Ccinf{\bar{E}}$ and $n\in \N$ the functions $f_n(\xbf)\coloneqq \chi(nx_2) f(\xbf)$ satisfy 
    \begin{align*}
        \lim_{n\to\infty} \ip{f_n-f}{\hb (f_n-f)}_{L^2(E)}=0.
    \end{align*}
\end{lemma}
\begin{proof}
    By definition there exists some $a>0$ such that $f$ is zero outside $U_a$. Hence by the mean value theorem there exists some $C>0$ such that for every $\x\in E$
    \begin{equation}\label{eq:f-estimate}
        \abs{f(\xbf)}\leq C \chi_{(-a,a)}(x_1) x_2,
    \end{equation}
    where $\chi_{(-a,a)}$ denotes the indicator function on the interval $(-a,a) \subset \R$.
   By Green's theorem it follows that
    \begin{align*}
    \ip{f_n-f}{\hb(f_n-f)}
     &\leq \int_E\Big[ \abs{\chi(nx_2)-1}^2[\abs{(-\im\partial_{1}+bx_2)f(\xbf)}^2+2\abs{\partial_{2}f(\xbf)}^2]\\
        &\phantom{\leq \int_E\Big[}+2n^2\abs{\chi'(n x_2)}^2\abs{f(\xbf)}^2\Big]\dd \xbf,
    \end{align*}
    and by Lebesgue's dominated convergence theorem the first term above converges to 0 when $n\to\infty$. Furthermore, an explicit calculation using \eqref{eq:f-estimate} shows that the second term above goes like $n^{-1}$ and the proof is over.
\end{proof}

The last part of Theorem~\ref{thm:1}\ref{thm:1Part1} is the content of the following lemma. 
\begin{lemma}
    The operator $H^E_b$ coincides with the Friedrichs extension $H^E_{b,F}$.
\end{lemma}
\begin{proof}
    Since $\Ccinf{E}$ is dense in $\mathcal{U}$ with respect to $\norm{\cdot}_\mathcal{U}=q_{b}^E(\cdot,\cdot)^{1/2}$ and
    \begin{equation*}
        Q_{b}^E(f,g)=\ip{f}{(\hb+W)g}_{L^2(E)} \qquad \forall \, f,g \,\in \Ccinf{E} \, ,
    \end{equation*}
    it follows from Lemma~\ref{lem:ipHb} and \eqref{eq:W-estimate} that $\Ccinf{\bar{E}}\subset \mathcal{U}$. Then, consider $f\in \Ccinf{\bar{E}}$ and $g\in \Ccinf{E}$, and let $f_n:=\chi_n f$, where $\chi_n$ is defined as in the proof of Lemma~\ref{lem:ipHb}. By the symmetry of $\hb+W$ and the $L^2$ convergence of $f_n$ to $f$ we get
    \begin{equation}
    \label{eq:aux1}
        q_{b}^E(g,f)=\ip{g}{(\hb+W) f}_{L^2(E)}+(1+M)\ip{g}{f}_{L^2(E)}.
    \end{equation}
    Since $\Ccinf{E}$ is dense in $\mathcal{U}$, the identity \eqref{eq:aux1} extends by continuity to all $g\in \mathcal{U}$. This shows that $\Ccinf{\bar{E}}\subseteq D(H_{b,F}^E)$ and that $H_{b,F}^E$ is an extension of $\hb+W$.
\end{proof}

{
\subsection{Proof of \texorpdfstring{Theorem~\ref{thm:1}\ref{thm:1Part2}}{Theorem 1.1(ii)}}
In order to prove Theorem~\ref{thm:1}\ref{thm:1Part2}, we first show \eqref{eq:HEexpansion} for $(\HEb+\lambda)^{-1}$, i.e.\ in the case when $W=0$, using the explicit structure of $(\HEb+\lambda)^{-1}$ given by \eqref{eq:HE_resolvent} and the estimate on the integral kernels provided in Appendix~\ref{s:intop}. Then, the general case is obtained using the second resolvent identity. 

\begin{lemma}\label{lem:hb123}
    Let $\Omega\subset \R$ be compact. There exist $\lambda_0>0$, $\delta>0$, and $C>0$ such that for all $\lambda\geq\lambda_0$ and $b\in \Omega$, 
    \begin{align}
        (h_b+\lambda)^{-1}(\xbf,\xbf')=\e^{\im b\phi(\xbf,\xbf')}(R_E(\xbf,\xbf')+\mathcal{K}_{b}(-\lambda)(\xbf,\xbf')) \label{eq:hbl1}
    \end{align}
    where $\mathcal{K}_{b}(-\lambda)$ is a smooth function of $b\in \Omega$ satisfying
\begin{equation}\label{eq:Kb}
    \abslr{\frac{\dd^m}{\dd b^m}\mathcal{K}_{b}(-\lambda)(\xbf,\xbf') }\leq C (16m^2+1)^m \lambda^{-1}\e^{-\delta\sqrt{\lambda}\abs{\xbf-\xbf'}}.
\end{equation}
\end{lemma}
\begin{proof}
Let $b$ and $\lambda$ be arbitrary. In the following we consider products of the form $S_b(-\lambda)(-T_b(-\lambda))^n$ for $n\geq1$ and their derivatives with respect to $b$. Since we want to isolate a phase factor $\e^{\im b\phi(\xbf,\xbf')}$, we define $\phi_s(\xbf,\ybf)\coloneqq-\frac{1}{2}(x_1y_2-x_2y_1)$
and note that (cf.\ \eqref{eq:phi})
\begin{equation}\label{eq:phis}
    \phi(\xbf,\ybf)+\phi(\ybf,\xbf')=\phi(\xbf,\xbf')+\phi_s(\xbf-\ybf,\ybf-\xbf').
\end{equation}
A key property of $\phi_s$ is that $ \abslr{\phi_s(\xbf,\ybf)}\leq \abs{\xbf}\abs{\ybf}$.

By \eqref{eq:Tlambda} and \eqref{eq:T-estimate} we may write
\begin{equation}
\label{eq:forT}
    T_b(-\lambda)(\xbf,\ybf)=\e^{\im b\phi(\xbf,\ybf)} t_b(-\lambda)(\xbf,\ybf),
\end{equation}
where $t_b(-\lambda)$ is differentiable with respect to $b$ and there exists some constant $C>0$ (depending on $\Omega$) such that for all $\lambda>1$, $m\in \N$, and $b\in \Omega$,
\begin{align}\label{eq:tb}
  \abslr{\frac{\dd^m }{\dd b^m}t_b(-\lambda)(\xbf,\ybf)}\leq C \e^{-\frac{\sqrt{\lambda}}{2}\abs{\xbf-\ybf}}.
\end{align}
In fact, the left hand side above is identically zero for $m\geq 3$, cf.\ \eqref{eq:Tlambda}.
Combining \eqref{eq:forT} and \eqref{eq:Slambda1} with the composition rule for the Peierls phase \eqref{eq:phis}, gives for every $n\in \N$, 
\begin{align*}
    [S_b(-\lambda)(-T_b(-\lambda))^n](\xbf,\xbf')= \e^{\im b\phi(\xbf,\xbf')}K_{b,n}(-\lambda)(\xbf,\xbf')\, ,
\end{align*}
where the kernel $K_{b,n}(-\lambda)$ is defined by
\begin{align*}
    K_{b,n}(-\lambda)(\xbf,\xbf')=(-1)^n\int_E\dots \int_E k_{b,n}(-\lambda)(\xbf,\ybf_1,\dots,\ybf_{n},\xbf')\dd \ybf_1\dots\ybf_{n}
\end{align*}
with 
\begin{align*}
    k_{b,n}(-\lambda)(\xbf,\ybf_1,\dots,\ybf_{n},\xbf')&=\e^{\im b\phi_s(\xbf-\ybf_1,\ybf_1-\ybf_2)+\dots+\im b\phi_s(\ybf_{n-1}-\ybf_{n},\ybf_{n}-\xbf')}\\
    &\quad \cdot R_E(\xbf,\ybf_1)t_b(-\lambda)(\ybf_1,\ybf_2)\dots t_b(-\lambda)(\ybf_{n},\xbf').
\end{align*}
Next we show that the derivatives of $k_{b,n}$ with respect to $b$ can be bounded by functions which are integrable on $E^n$ for $b\in\Omega$. This is a consequence of the fact that all kernels appearing in $k_{b,n}$ decay exponentially away from the diagonal and thus the derivatives of the phase factor in $k_{b,n}$ can be bounded uniformly on $E^n$. Specifically, by \eqref{eq:tb}, \eqref{eq:S-estimate}, and the fact that $x^n \e^{-ax}\leq n^n a^{-n}$ for all $x\geq 0$, it follows that there exists $C>0$ such that for all $b\in \Omega$, $\lambda>1$ and $m\in \N$ we have
\begin{align*}
    \abslr{\frac{\dd^m}{\dd b^m}k_{b,n}(-\lambda)(\xbf,\ybf_1,\dots,\ybf_{n-1},\xbf') }&\leq C^{n+1}(16m^2+1)^m (1+\abs{\ln(\sqrt{\lambda}\abs{\xbf-\ybf_1})})\\
    &\quad \cdot \e^{-\frac{\sqrt{\lambda}}{4}(\abs{\xbf-\ybf_1}+\dots+\abs{\ybf_{n}-\xbf'})}.
\end{align*} 
Hence, by Lebesgue's dominated convergence theorem,
\begin{equation*}
   \abslr{ \frac{\dd^m}{\dd b^m}K_{b,n}(-\lambda)(\xbf,\xbf') }\leq  \frac{C^{n+1}}{\lambda^{n}}(16m^2+1)^m \e^{-\frac{\sqrt{\lambda}}{8}\abs{\xbf-\xbf'}},
\end{equation*}
where $C$ is some generic constant. Taking $\lambda$ sufficiently large we obtain from \eqref{eq:HE_resolvent} that 
\begin{align*}
    (h_b+\lambda)^{-1}(\xbf,\xbf')&=\e^{\im b\phi(\xbf,\xbf')}\Big(R_E(\xbf,\xbf')+\sum_{n=1}^\infty K_{b,n}(-\lambda)(\xbf,\xbf')\Big)\nonumber\\
    &\eqqcolon\e^{\im b\phi(\xbf,\xbf')}(R_E(\xbf,\xbf')+\mathcal{K}_{b}(-\lambda)(\xbf,\xbf')) 
\end{align*}
where $\mathcal{K}_{b}(-\lambda)$ is a smooth function of $b\in \Omega$ satisfying \eqref{eq:Kb}.
\end{proof}
Repeating now the arguments in the proof of Lemma~\ref{lem:hb123} but with $S_b(-\lambda)$ replaced by $\overline{W}S_b(-\lambda)$ and using \eqref{eq:1divx-estimate} instead of \eqref{eq:S-estimate} we obtain the following result. 
\begin{lemma}\label{lem:whb123}
    Let $\Omega\subset \R$ be compact. There exist $\lambda_0>0$, $\delta>0$ and $C>0$ such that for all $\lambda\geq\lambda_0$ and $b\in \Omega$, 
    \begin{align}\label{eq:Whbl1}
        -\overline{W}(h_b+\lambda)^{-1}(\xbf,\xbf')&=\e^{\im b\phi(\xbf,\xbf')}(\mathcal{W}(-\lambda)(\xbf,\xbf')+\tilde{\mathcal{K}}_{b}(-\lambda)(\xbf,\xbf'))
    \end{align}
    where $\e^{\im b\phi(\xbf,\xbf')}\mathcal{W}(-\lambda)(\xbf,\xbf')$ denotes the integral kernel of $-\overline{W}S_b(-\lambda)$ and the integral kernel $\tilde{\mathcal{K}}_{b}(-\lambda)$ is a smooth function of $b\in \Omega$ satisfying
    \begin{equation}\label{eq:Kbtilde}
        \abslr{\frac{\dd^m}{\dd b^m}\tilde{\mathcal{K}}_{b}(-\lambda)(\xbf,\xbf') }\leq C(16m^2+1)^m\e^{-\delta \sqrt{\lambda}\abs{\xbf-\xbf'}}.
    \end{equation}    
\end{lemma}
Using Lemma~\ref{lem:hb123} and Lemma~\ref{lem:whb123} we are now able to finish the proof of Theorem~\ref{thm:1}.
\begin{proof}[Proof of Theorem~\ref{thm:1}\ref{thm:1Part2}]
    If $\lambda>0$ is sufficiently large, then the second resolvent identity together with \eqref{eq:1divx-estimate} imply that 
\begin{equation}\label{eq:resolvent1}
    (\HE+\lambda)^{-1}=(\HEb+\lambda)^{-1}(1+\overline{W}(\HEb+\lambda)^{-1})^{-1}.
\end{equation}
Therefore, to extend the result in Lemma~\ref{lem:hb123} to the full resolvent $(\HE+\lambda)^{-1}$ we have to consider products of the form $(\HEb+\lambda)^{-1}(-\overline{W}(\HEb+\lambda)^{-1})^{n}$ for $n\geq 1$.  

By repeating the arguments used in the proof of Lemma~\ref{lem:hb123}, but using instead \eqref{eq:hbl1} and \eqref{eq:Whbl1} together with the estimates \eqref{eq:Kb} and \eqref{eq:Kbtilde} we obtain that   
\begin{equation}\label{eq:hbwhb}
    (\HEb+\lambda)^{-1}(-\overline{W}(h_b+\lambda)^{-1})^n(\xbf,\ybf)=\e^{\im b\phi(\xbf,\ybf)}(\mathcal{V}_n(-\lambda)(\xbf,\ybf)+\tilde{\mathcal{K}}_{b,n}(-\lambda)(\xbf,\ybf)),
\end{equation}
where $\e^{\im b\phi(\xbf,\ybf)} \mathcal{V}_n(-\lambda)(\xbf,\ybf)$ is the kernel of $S_b(-\lambda)(-\overline{W}S_b(-\lambda))^n$ and $\tilde{\mathcal{K}}_{b,n}(-\lambda)$ is a smooth function of $b\in \Omega$ which satisfies 
\begin{align*}
    \abslr{\frac{\dd^m }{\dd b^m}\tilde{\mathcal{K}}_{b,n}(-\lambda)(\xbf,\ybf)}\leq C(16m^2+1)^m (n+2)^m C_\lambda^{n} \e^{-\delta\sqrt{\lambda}\abs{\xbf-\ybf}}
\end{align*}
where $C_\lambda \to 0$ as $\lambda \to \infty$. For the kernel $\mathcal{V}_n(-\lambda)$ one obtains a similar estimate with almost the same arguments. The only change is that both $S_b(-\lambda)$ and $\overline{W}S_b(-\lambda)$ have kernels that are not bounded on $E\times E$. However, this problem can be circumvented by using Hölder's inequality together with the specific estimates in \eqref{eq:S-estimate} and \eqref{eq:1divx-estimate}. Hence we have 
\begin{equation*}
    \abslr{\frac{\dd^m}{\dd b^m}\mathcal{V}_n(-\lambda)(\xbf,\xbf') }\leq C(16m^2+1)^m\tilde{C}_\lambda^{n}\e^{-\delta \sqrt{\lambda}\abs{\xbf-\ybf}}\, ,
\end{equation*}
where $\tilde{C}_\lambda$ goes to $0$ for $\lambda\to\infty$. By choosing $\lambda$ sufficiently large and summing \eqref{eq:hbwhb} from $n=1$ to $\infty$ gives a kernel which is smooth as a function of $b\in \Omega$. By \eqref{eq:resolvent1} the proof is complete.
\end{proof}

\section{Smoothness of bulk and edge current densities}
\label{sec:regularity}
The aim of this section is to prove Proposition~\ref{prop:2} and Theorem~\ref{thm:smoothnes}. The proofs are based on the explicit integral kernel estimates in Appendix~\ref{s:intop}.
We start with a simple but fundamental lemma showing norm exponential localization properties of the operators we deal with.

\begin{lemma}\label{lem:eRe}
    Let $\mathcal{R}_k=(-\overline{W})^k\Rb$ for $k\in \{0,1\}$. Then, for any $\lambda>0$ sufficiently large there exists $\epsilon >0$, such that for any
    \[Y\in \{S_b(-\lambda),T_b(-\lambda),(1+T_b(-\lambda))^{-1}, \mathcal{R}_1,(1-\mathcal{R}_1)^{-1},\mathcal{R}_2,(1-\mathcal{R}_2)^{-1}\}\]
    the operator $\e^{-\epsilon \abs{X}}Y\e^{\epsilon\abs{X}}$ is bounded on $L^2(E)$.
\end{lemma}
\begin{proof}
    The results follow from simple calculations using the explicit estimates in \eqref{eq:S-estimate}, \eqref{eq:T-estimate}, \eqref{eq:Rb-estimate}, and \eqref{eq:1divx-estimate}, together with the triangle inequality and the Schur test.
\end{proof}

\subsection{Proof of \texorpdfstring{Proposition~\ref{prop:2}}{Proposition 1.4}}

    In the following we use the shorthand notation $R_{b}^E=\RbE$ and $\mathcal{R}_k=(-\overline{W})^k\Rb$, for $k\in \{0,1\}$. Let $\lambda>0$ be large enough such that the statements in Lemma~\ref{lem:eRe} hold true for some $\epsilon>0$. Then, it is possible to show that for any $N\in \N$,
    \begin{equation}\label{eq:RboEN}
        (R_{b}^E)^N=\mathcal{R}_0\sum_{k_j\in \{0,1\}}\mathcal{R}_{k_1}\cdots \mathcal{R}_{k_{N-1}} B_{(k_1,\dots, k_{N-1})},
    \end{equation}
   where the sum in \eqref{eq:RboEN} goes over all the possible $N-1$-tuples constructed from the set $\{0,1\}$, $B_{(k_1,\dots k_{N-1})}$ is a bounded operator for which $\e^{-\epsilon\abs{X}}B_{(k_1,\dots k_{N-1})}\e^{\epsilon\abs{X}}$ is also bounded for all $N-1$-tuples $(k_1,\dots ,k_{N-1}$), and the sum is defined as $(1-\mathcal{R}_1)^{-1}$ when $N=1$. Let us show this for $N=3$ as the cases for $N\neq 3$ follow exactly the same strategy. Starting from \eqref{eq:resolvent1}, that is  $R_{b}^E=\mathcal{R}_0(1-\mathcal{R}_1)^{-1}$. and expanding $(1-\mathcal{R}_1)^{-1}$ to first and second order we obtain
    \begin{align*}
        (R_{b}^E)^3&=\mathcal{R}_0(1+\mathcal{R}_1+\mathcal{R}_1^2(1-\mathcal{R}_1)^{-1})\mathcal{R}_0(1+\mathcal{R}_1(1-\mathcal{R}_1)^{-1})\mathcal{R}_0(1-\mathcal{R}_1)^{-1}\\
        &=\mathcal{R}_0^3B_{(0,0)}+\mathcal{R}_0^2\mathcal{R}_1B_{(0,1)}+\mathcal{R}_0\mathcal{R}_1\mathcal{R}_0B_{(1,0)}+\mathcal{R}_0\mathcal{R}_1^2B_{(1,1)}
    \end{align*}
    for some bounded operators $B_{(0,0)},B_{(0,1)},B_{(1,0)},B_{(1,1)}$, which by Lemma~\ref{lem:eRe} are such that $\e^{-\epsilon\abs{X}}B_{(k_i,k_j)}\e^{\epsilon\abs{X}}$ is bounded for all $2$-tuples $(k_i,k_j)\in \{0,1\}^2$.

    Since $\lambda$ and $b$ are fixed, we use the shorthand notation $S_b(-\lambda)\equiv S$ and $T_b(-\lambda)\equiv T$  for the remaining part of this proof. By \eqref{eq:HE_resolvent} we can express the resolvent of the edge Hamiltonian in terms of $S$ and $T$, that is $\mathcal{R}_0=S(1+T)^{-1}$ and $\mathcal{R}_1=-\overline{W}S(1+T)^{-1}$, thus obtaining
    \begin{equation*}
        \mathcal{R}_0\mathcal{R}_{k_1}\dots \mathcal{R}_{k_{N-1}}=S(1+T)^{-1}(-\overline{W})^{k_1}S(1+T)^{-1}\dots (-\overline{W})^{k_{N-1}}S(1+T)^{-1}.
    \end{equation*}
    Recall that $k_i \in \{0,1\}$, for all $0\leq i \leq N$.
    By writing $(1+T)^{-1}=1-T+\dots +(-T)^{n}+(-T)^{n+1}(1+T)^{-1}$ and using the same argument that gave \eqref{eq:RboEN}, we have that $ \mathcal{R}_0\mathcal{R}_{k_1}\dots \mathcal{R}_{k_{N-1}}$ can be written as a sum of terms of the form 
    \begin{equation}
    \label{eq:aux2}
       S(-T)^{a_0}(-\overline{W})^{k_1}S(-T)^{a_1}(-\overline{W})^{k_2}S(-T)^{a_2}\cdots(-\overline{W})^{k_{n}}S)(-T)^{a_{n}}B
    \end{equation}
    where $n\in \{0,\dots,N-1\}$ and $a_j\in \{0,\dots,N-1\}$ with $n+a_0+\dots+a_n=N-1$ (when $n=0$ the term is of the form $S(-T)^{N-1}$) and $B$ is some bounded operator. By Lemma~\ref{lem:eRe} and \eqref{eq:RboEN}, $\e^{-\epsilon\abs{X}}B\e^{\epsilon\abs{X}}$ is bounded. Note that $n$ counts the number of times an operator of the form $(-\overline{W})^{k_j}S$ appears in \eqref{eq:aux2}. Likewise, $a_0+\dots+a_n$ counts the number of times the operator $T$ appears in \eqref{eq:aux2}. 
    The key observation is that every term in the sum in \eqref{eq:RboEN} can be written as a product of the form 
    \begin{equation*}
        SY_{2}\dots Y_{N} B
    \end{equation*} 
    where $Y_j$ is either $S$, $T$, or $\overline{W}S$ and $B$ is a bounded operator, such that $\e^{-\epsilon\abs{X}}B\e^{\epsilon\abs{X}}$ is also bounded. 
    
    By the proof of Lemma~\ref{lem:HX} below, it follows that $\im[\HE,X_j]S=2\overline{P_j}S$. Hence, it only remains to show that for any $m\in \N$, there exists $N$ sufficiently large, such that any operator of the form
    \begin{equation*}
        gY_{1}Y_{2}\dots Y_{N} \e^{\epsilon\abs{X}},
    \end{equation*} 
    where $Y_i$, $i \in \{1,\dots,N\}$, is either $S$, $T$, $\overline{W}S$, or $\overline{P_j}S$, maps $L^2(E)$ continuously into $C^m(E)$. However, this immediately follows from Lemma~\ref{lem:stws}. \hfill \qedsymbol
 
\subsection{Proof of \texorpdfstring{Theorem~\ref{thm:smoothnes}}{Theorem 1.5}}

   We only consider $g_1F(\HE)g_2$ since the other cases follow exactly the same line of argument.
    
    Let $\lambda>0$ be large enough that Proposition~\ref{prop:2} holds and let $m\in \N$ be arbitrary. Choose $N$ and $\epsilon$ as in Proposition~\ref{prop:2}. By the functional calculus we may write 
    \begin{align*}
        g_1F(\HE)g_2&=\Big(g_1 (\HE+\lambda)^{-N}\e^{\epsilon\abs{X}}\Big)\\
        &\quad \cdot\Big(\e^{-\epsilon\abs{X}}(\HE+\lambda)^{-1}F(\HE)(\HE+\lambda)^{2N+1}\Big)\\
        &\quad \cdot\Big((\HE+\lambda)^{-N}g_2\Big)\\
        &\eqqcolon ABC,
    \end{align*}
    where $A$ and $C$ are bounded operators. Next we argue that $B$ is a Hilbert-Schmidt operator. Since $F$ decays faster than any polynomial, $F(\HE)(\HE+\lambda)^{2N+1}$ is bounded by the functional calculus. Thus we only have to show that $\e^{-\epsilon\abs{X}}(\HE+\lambda)^{-1}$ is a Hilbert-Schmidt operator. By the proof of Proposition~\ref{prop:2}, $(\HE+\lambda)^{-1}=S_b(-\lambda)\mathcal{B}$ for some bounded operator $\mathcal{B}$, and thus the estimate \eqref{eq:S-estimate} implies that $\e^{-\epsilon\abs{X}}(\HE+\lambda)^{-1}$ is a Hilbert-Schmidt operator. Hence $B$ is a Hilbert-Schmidt operator with a kernel $B(\cdot,\cdot)\in L^2(E\times E)$ given by
    \begin{equation*}
        B(\xbf,\ybf)=\sum_{j,k=1}^\infty c_{jk}\psi_{j}(\xbf)\overline{\psi_k(\ybf)},
    \end{equation*}    
    where $\{\psi_n\}_{n\geq 1}$ is an orthonormal basis in $L^2(E)$ and $\sum \abs{c_{jk}}^2<\infty$. This shows that $g_1F(\HE)g_2$ is an integral operator with kernel
    \begin{equation*}
        (g_1F(\HE)g_2)(\xbf,\ybf)=\sum_{j,k=1}^\infty c_{jk}(A\psi_{j})(\xbf)\overline{(C^*\psi_k)(\ybf)}.
    \end{equation*}
    For any $n\in \N$, Proposition~\ref{prop:2} implies that the function 
    \begin{equation*}
        K_n(\xbf,\ybf)=\sum_{j,k=1}^n c_{jk}(A\psi_{j})(\xbf)\overline{(C^*\psi_k)(\ybf)}
    \end{equation*}
    is $C^m(E\times E)$. Furthermore, combining Proposition~\ref{prop:2} and Lemma~\ref{lem:L2toLinf} shows that for any $\alpha\in \N_0^{2}$ with $\abs{\alpha}\leq m$, both $\partial^\alpha A$ and $\partial^\alpha C^*$ are integral operators with kernels satisfying $ (\partial^\alpha A)(\xbf,\cdot),(\partial^\alpha C^*)(\xbf,\cdot)\in L^2(E)$ for every $\xbf \in L^2(E)$ and with $L^2(E)$ norms uniformly bounded in $\xbf$. Using this together with Parseval's identity gives that there exists $C>0$ such that for all $\xbf\in E$,
    \begin{equation*}
        \sum_{j=1}^\infty \abs{(\partial^\alpha A\psi_{j})(\xbf)}^2=\sum_{j=1}^\infty \abs{\ip{\overline{\partial^\alpha A(\xbf,\cdot)}}{\psi_j}_{L^2(E)}}^2= \norm{\partial^\alpha A(\xbf,\cdot)}_{L^2(E)}^2\leq C.
    \end{equation*}
    Similar bound holds when replacing $\partial^\alpha A$ with $\partial^\alpha C^*$. Using such estimates we obtain for any $\alpha,\beta\in \N_0^{2}$ with $\abs{\alpha}+\abs{\beta}\leq m$, that 
    \begin{equation}
        \label{eq:auxHS}
    \begin{aligned}
        &\abslr{\sum_{j,k=n+1}^\infty c_{jk}(\partial^\alpha A\psi_{j})(\xbf)\overline{\partial^\beta C^*\psi_k(\ybf)}}^2\\
        &\qquad\leq \Big(\sum_{j,k=n+1}^\infty \abs{c_{jk}}^2\Big) \Big(\sum_{j=n+1}^\infty \abs{(\partial^\alpha A\psi_{j})(\xbf)}^2\Big)\Big(\sum_{k=n+1}^\infty \abs{(\partial^\beta C^*\psi_{j})(\ybf)}^2\Big)\\
        &\qquad\leq c\sum_{j,k=n+1}^\infty \abs{c_{jk}}^2,
    \end{aligned}
    \end{equation}
    where $c$ is some constant. Hence, $\partial_1^{\alpha}\partial_2^\beta K_n(\xbf,\ybf)$ converges uniformly for $\abs{\alpha}+\abs{\beta}\leq m$ and thus the integral kernel of $g_1F(\HE)g_2$ is $C^{m}(E)$. Since $m$ was chosen arbitrarily $g_1F(\HE)g_2$ has a smooth kernel.    
\hfill \qedsymbol

\section{Comparison of bulk and edge current densities}
\label{sec:GPT}
This section is devoted to the proof of Theorem~\ref{thm:currents}.
Let us start by setting up the general framework of geometric perturbation theory as done in \cite{CorneanMoscolariTeufel}.

Fix $\ell>2$ and consider the set $\Xi(t):=\{\x \in E \, | \, \mathrm{dist}(\x,\partial E) \leq t \sqrt{\ell} \}$, $t>0$.  Then let $0\leq \eta_0, \eta_\ell \leq 1$ be two smooth non-negative functions only depending on $x_2$ such that $\eta_0(\x)+\eta_\ell(\x)=1$ for every $\x \in E$. Moreover, we assume that 

\begin{equation*}
\begin{aligned}
&\mathrm{supp}(\eta_0) \subset \Xi(2) ,\\
&\mathrm{supp}(\eta_\ell) \subset E \setminus \Xi(1) , \\
&\| \partial^{n}_2\eta_i \|_{\infty} \simeq \ell^{-\frac{n}{2}},\quad n\geq 1, \quad i \in \left\{0,\ell\right\} .
\end{aligned}
\end{equation*}
We now introduce another couple of  non-negative functions  $0\leq \widetilde{\eta}_0, \widetilde{\eta}_\ell \leq 1$ again only depending on $x_2$, with the properties:
\begin{equation}
\label{eq:defEta2}
\begin{aligned}
& \mathrm{supp}(\widetilde{\eta}_0) \subset  \Xi\Big(\frac{11}{4}\Big), \\
& \mathrm{supp}(\widetilde{\eta}_\ell) \subset E \setminus \Xi\Big(\frac{1}{4}\Big), \\
& \widetilde{\eta}_i \eta_i = \eta_i , \quad i \in \left\{0,\ell\right\} ,  \\
& \dist\left(\mathrm{supp}(\partial_2 \widetilde{\eta}_i), \mathrm{supp}(\eta_i) \right) \simeq \sqrt{\ell} , \\
&\| \partial^{n}_2{ \widetilde{\eta_i}} \|_{\infty} \simeq \ell^{-\frac{n}{2}},\quad n\geq 1, \quad i \in \left\{0,\ell\right\} .
\end{aligned}
\end{equation}
The functions $\widetilde{\eta}_0, \widetilde{\eta}_\ell$ are a sort of stretched version of ${\eta}_0$ and ${\eta}_\ell$, such that the supports of the derivative of $\eta_i$ and $\widetilde{\eta}_i$ are disjoint.

{
\begin{lemma}\label{lem:HX}
    Let $j\in \{1,2\}$ and let $z$ be in the resolvent set of $\HE$. Then $\im[\HE,X_j] (\HE-z)^{-1}$ is bounded on $L^2(E)$ and for any bounded smooth function $\eta\in \Cinf{E}$ which does not depend on $x_j$,
    \begin{equation*}
        \eta\im[\HE,X_j] (\HE-z)^{-1}=\im[\HE,X_j]\eta (\HE-z)^{-1}.
    \end{equation*}
\end{lemma}
\begin{proof}
    First we consider the case $z=\lambda>0$ for $\lambda$ sufficiently large. For any $f\in\M$ (cf.\ Lemma~\ref{lem:MinDHb}), a straightforward computation (mimicking the proof of Lemma~\ref{lem:MinDHb}) shows that $\im [\HE,X_j] f=2\overline{P_j}f$. Hence by arguing as in the proof of Proposition~\ref{prop:1} and applying Lemma~\ref{lem:SM} and \eqref{eq:1divx-estimate} we obtain
    $\im[\HE,X_j] S_b(-\lambda)= 2 \overline{P_j} S_b(-\lambda)$. Combining this with \eqref{eq:HE_resolvent} shows that $\im [\HE,X_j](\HEb+\lambda)^{-1}=2\overline{P_j}(\HEb+\lambda)^{-1}.$
    Now the result follows from \eqref{eq:resolvent1}, and the fact that $\eta$ does not depend on $x_j$. In the general case we use the first resolvent identity,
    \begin{equation}\label{eq:firstresolvent}
        (\HE-z)^{-1}=(\HE+\lambda)^{-1}+(z+\lambda) (\HE+\lambda)^{-1}(\HE-z)^{-1},  
    \end{equation}
    from which the result immediately follows. 
\end{proof}
}
The space $L^2(E)$ can be canonically identified with a subset of $L^2(\R^2)$, where the identification operator $I_{E}: L^2(E)  \to L^2(\R^2)$ is given by 
\begin{equation*}
    I_{E}(\psi):= \chi_{E} \psi \in L^2(\R^2) .
\end{equation*}
\begin{proposition}
\label{prop:bulkEdgeHamt}

The multiplication operator by $\widetilde{\eta}_\ell$ maps the domain of $\HB$ into the domain of $\HE$. Moreover, for every $\psi$ in the domain of $\HB$ we have 
\begin{equation*}
    \HE \, \tilde{\eta}_{\ell} \psi =  \HB  \, \tilde{\eta}_{\ell} \psi .
\end{equation*}
\end{proposition}
\begin{proof}
Consider $\psi \in D(\HB)$. Then, there exists a sequence $\varphi_n\in C^{\infty}_c(\R^2)$ such that $\varphi_n \to \psi$ and $\HB\varphi_n \to \HB\psi$ in $L^2(\R^2)$. Moreover, we have that $\tilde{\eta}_{\ell}\varphi_n  \in C^{\infty}_c(\bar{E})$, $\tilde{\eta}_{\ell}\varphi_n \to \tilde{\eta}_{\ell}\psi$ and $\HB\tilde{\eta}_{\ell}\varphi_n \to \HB\tilde{\eta}_{\ell}\psi$ in $L^2(\R^2)$, which follows from \eqref{eq:explicitHb}. By definition of $\HE$, we have that 
\begin{equation*}
	\HE\left(\tilde{\eta}_{\ell}\varphi_n\right)=(\hb+W )\left(\tilde{\eta}_{\ell}\varphi_n\right)=\HB \left(\tilde{\eta}_{\ell}\varphi_n\right)\, .
\end{equation*}
Moreover, from the inclusion $L^2(E) \subset L^2(\R)$ we have that $(\hb+W )\left(\tilde{\eta}_{\ell}\varphi_n\right)$ is a Cauchy sequence in $L^2(E)$ and since $\HE=\overline{\hb+W}$, the proof is concluded. 
\end{proof}
Note that a direct consequence of Proposition~\ref{prop:bulkEdgeHamt} is the identity 
\begin{equation}\label{eq:HBX-HEX}
    \im [\HB,X_1]\widetilde{\eta}_\ell=\im (\HB\widetilde{\eta}_\ell X_1-X_1\HB\widetilde{\eta}_\ell)=\im [\HE,X_1]\widetilde{\eta}_\ell
\end{equation}
Define, for any $z\in \rho(\HB)\cap\rho(\HE)$, the bounded operator in $L^2(E)$:
\begin{equation*}
    U_{\ell}(z):=\widetilde{\eta}_{\ell} \left(\HB - z \right)^{-1} \eta_{\ell}  +   \widetilde{\eta}_0 \left( \HE - z \right)^{-1} \eta_0.
\end{equation*}
Then, for $z$ as above, we have
\begin{equation*}
    \left(\HE - z \right)U_{\ell}(z)=1 + W_{\ell}(z) \, ,
\end{equation*}
where $W_{\ell}(z)$ is the bounded operator given by 
\begin{equation}\label{eq:Wlomega}
    \begin{aligned}
        W_{\ell}(z)&:= \left(-2 \iu \nabla \widetilde{\eta}_\ell \cdot \left(-\iu\nabla - \mathcal{A} - b A\right) - (\Delta \widetilde{\eta}_\ell) \right)   \left(\HB - z \right)^{-1} \eta_{\ell}  \\
        &\qquad+ \left(-2 \iu \nabla \widetilde{\eta}_0 \cdot \left(-\iu\nabla - \mathcal{A} - b A\right) - (\Delta \widetilde{\eta}_0) \right)   \left( \HE - z \right)^{-1} \eta_{0} .
        \end{aligned}        
\end{equation}
Therefore, the resolvent of the edge Hamiltonian obeys the identity:
\begin{equation}
\label{gpt1}
\left(\HE - z \right)^{-1}=U_{\ell}(z)  - \left(\HE - z \right)^{-1}W_{\ell}(z),
\end{equation}
which will be crucial in the proof of Theorem~\ref{thm:currents}.

Another important ingredient for the proof is the Helffer-Sjöstrand formula \cite{HelfferSjostrand1989,davies1995spectral} which we recall here. By following the strategy of \cite{CorneanFournaisFrankHelffer,CorneanMoscolariTeufel}, we can write $F(H^{\;\; / E}_{b})$ (for $F\in \Schwartsfunct{\R}$) as
\begin{equation}\label{dc22}
F(H^{\;\; / E}_{b})=-\frac{1}{\pi} \int_{\mathcal{D}}  \bar{\partial} {F_N}(z) (H^{\;\; / E}_{b}-z)^{-1} \mathrm{d}z_1\mathrm{d}z_2\, , \quad z=z_1+\iu z_2,
\end{equation}
where $\mathrm{Re}z=z_1, \mathrm{Im}z=z_2 \in \R$, $\mathrm{d}z_1\mathrm{d}z_2$ is the Lebesgue measure of $\C\cong \R \times \R$, $\mathcal{D}:=\R \times [-1,1]$, and $F_N$ is an almost analytic extension of $F$ constructed as follows: Let $0\leq g(y)\leq 1$ with $g\in \Ccinf{\R}$ such that $g(y)=1$ if $|y|\leq 1/2$ and $g(y)=0$ if $|y|>1$. Fix some $N\geq 2$ and define 
	\begin{equation*}
	F_N(z_1+\iu z_2) := g(z_2) \sum_{j=0}^N \frac{1}{j!} \frac{\partial^j F}{\partial {z_1}^j}(z_1) (\iu z_2)^j.
	\end{equation*}
We note the following properties which are easily verified:
\begin{enumerate}[label=(\roman*)]
	\item $F_N(x)=F(x) \qquad \textup{for all } \, x \in \R$;
	\item $\mathrm{supp}(F_N) \subset \mathcal{D}$;
	\item \label{prop3HS} there exists a positive constant $C_N$ (dependent on $N$) such that
	\begin{equation*}
        \left|\bar{\partial} {F_N}(z)\right| \leq C_N \frac{|z_2|^N}{\langle z_1\rangle^N} \, .
    \end{equation*}
\end{enumerate}
Property \ref{prop3HS} will be crucial in the proof of Theorem~\ref{thm:currents}.

\subsection{Proof of \texorpdfstring{Theorem~\ref{thm:currents}}{Theorem 1.7}}

Let us start by considering the operator associated to the current density in the upper half-plane. Recall that we use the notation $F$ for a Schwartz function which equals $F_{\mathrm{FD}}$ on the set $\left[-1+\inf \left(\sigma(\HB)\right),+\infty\right)$. Furthermore, we use the shorthand notation $\RB(z)=(\HB-z)^{-1}$ and $\RE(z)=(\HE-z)^{-1}$, $z\in \C$. In order to calculate the difference 
\[\im[\HE,X_1]F(\HE)-\chi_{E} \im [\HB,X_1]F(\HB) \chi_{E}\, ,\]
we apply the Helffer-Sjöstrand formula \eqref{dc22} to obtain 
\begin{align*}
        &\im[\HE,X_1]F(\HE)-\chi_{E} \im [\HB,X_1]F(\HB) \chi_{E}
        \\
        &\qquad=-\frac{1}{\pi} \int_{\mathcal{D}}  \bar{\partial} {F_N}(z) \Big(\im[\HE,X_1]\RE(z)- \chi_{E}\im[\HB,X_1]\RB(z)\chi_{E} \Big)\mathrm{d}z_1\mathrm{d}z_2.
\end{align*}
Thus everything is reduced to the analysis of $\HER(z)-\chi_E\HR(z)\chi_E$ where we denote $\HR(z)\coloneqq\im [\HB,X_1]\RB(z)$ and $\HER(z)\coloneqq\im [\HE,X_1]\RE(z)$.

By using geometric perturbation theory as described above, we can write the resolvent of the edge Hamiltonian as a sum of terms for which we can control the localization properties in the direction perpendicular to the boundary, cf.\ \eqref{eq:Wlomega}. Specifically, combining \eqref{gpt1}, that is 
\begin{equation*}
    \RE(z)=U_{\ell}(z)  - \RE(z)W_{\ell}(z),
\end{equation*}
    with \eqref{eq:HBX-HEX}, Lemma~\ref{lem:HX}, and its obvious counterpart for $\HB$ (cf.\ Remark~\ref{rem:HB}) we obtain
\begin{equation*}
    \HER(z)=\widetilde{\eta}_{\ell}\HR(z)\eta_{\ell}  +   \widetilde{\eta}_0 \HER(z) \eta_0-\HER(z) W_{\ell}(z).
\end{equation*}
By using also that $\chi_{E}=\eta_\ell +\eta_0$ we have
    \begin{equation}
    \label{eq:TClass1}
    \begin{aligned}
    \HER(z)  - \chi_{E}\HR(z)\chi_{E}&=(\widetilde{\eta}_{\ell}-\eta_\ell) \HR(z) \eta_{\ell}  - \eta_0 \HR(z) \eta_{\ell} \\ 
    &\quad-\chi_{E}\HR(z) \eta_{0} 
    +\widetilde{\eta}_0 \HER(z) \eta_0 - \HER(z)W_{\ell}(z) .
    \end{aligned}
    \end{equation}
    The first four terms can be treated in the same way since they all contain at least one factor which is localized near the boundary. In the following we will consider only the fourth term. Afterwards, we consider the last term which requires a separate analysis.
    
    In order to deal only with regular kernels, we employ the following strategy, borrowed from \cite[Proposition 1.10]{CorneanMoscolariTeufel}.  Whenever we have an operator family $A(z)$ which is analytic on $\mathcal{D}$ and with a polynomial growth in $z_1$ at infinity, then 
      \begin{equation*}
	\int_{\mathcal{D}}  \bar{\partial} {F_N}(z) A(z)\mathrm{d}z_1\mathrm{d}z_2=0.
\end{equation*}
    Let $\lambda>0$ be large enough such that Proposition~\ref{prop:2} holds. Since $(\HE+\lambda)^{-1}$ is clearly analytic on $\mathcal{D}$, applying the first resolvent identity \eqref{eq:firstresolvent} repeatedly gives
    \begin{equation}
    \label{eq:GreenthmInfty}
    \begin{aligned}
     &\int_{\mathcal{D}}  \bar{\partial} {F_N}(z)  \widetilde{\eta}_0\HER(z)\eta_0  \mathrm{d}z_1\mathrm{d}z_2 \\
     &= \int_{\mathcal{D}} \bar{\partial} F_N(z)\widetilde{\eta}_0  (\lambda+z)^n \im\left[\HE,X_1\right] \left( \HE +\lambda \right)^{-n} \left( \HE -z \right)^{-1}\eta_0\mathrm{d}z_1\mathrm{d}z_2  .
     \end{aligned}
    \end{equation}
    for any $n\in \N$. By making minor changes to the proof of Theorem~\ref{thm:smoothnes} it can be shown that $\im[\HE,X_1] ( \HE+ \lambda)^{-n} ( \HE -z)^{-1}$ is an integral operator with a kernel that becomes increasingly regular as $n$ goes to infinity. It suffices to choose $n$ such that the kernel becomes continuous on $E\times E$. Then the localization due to $\eta_0$ implies that the kernel of the operator in \eqref{eq:GreenthmInfty} is identically zero far from the boundary. 
    Note that for the terms in \eqref{eq:TClass1} involving $\HR(z)$ one uses an analogue of Theorem~\ref{thm:smoothnes}, cf.\ Remark~\ref{rem:HB}, to reach the same conclusion.

    Let us treat the last term in \eqref{eq:TClass1}. Since $N$ is arbitrary, it suffices to prove that for all $M\in \N$, there exists $N\in \N$ and a constant $C_{\ell,N,M}$ such that 
\begin{equation*}
	    \sup_{x_1 \in \R} x_2^{M} \abslr{\int_{\mathcal{D}} \bar{\partial} F_{N}(z) \left( \HER(z) W_{\ell}(z)\right)((x_1,x_2),(x_1,x_2))\mathrm{d}z_1\mathrm{d}z_2}\leq C_{\ell,N,M}.
\end{equation*}    
    Let us consider the diagonal of the kernel of $x_2^M\HER(z)W_{\ell}(z)$. Writing $x_2^M=[(x_2-y_2)+y_2]^M$ and expanding with the binomial theorem it follows that we have to consider kernels of the form
\begin{equation*}
	     K_{M,m}(\xbf,\xbf)=\int_{E} (x_2-y_2)^{m} (y_2)^{M-m} \left( \HER(z) \right)(\x;\y) (W_{\ell}(z))(\y;\x) \dd \ybf,
\end{equation*}    
    for $m\in \{0,\dots,M\}$.
    Define for any $z\in \C$, $\zeta\coloneqq \jnorm{\mathrm{Re}(z)}\abs{\mathrm{Im}(z)}^{-1}$. By using that $W_{\ell}(z)(\ybf,\xbf)=0$ for all $y_2>\frac{11}{4}\sqrt{\ell}$ (cf.\ \eqref{eq:defEta2} and \eqref{eq:Wlomega}) together with \cite[Lemma A.4]{CorneanMoscolariTeufel} (which follow from the estimates \eqref{eq:resolvent-estimate1} and \eqref{eq:resolvent-estimate2}), it follows that there exist $C,C_{\ell,M}, \delta_0> 0$ and $a \geq 1$ such that 
    \begin{equation*}
    \left|(y_2)^{M-m}W_{\ell}(z)(\y;\x)\right| \leq C \ell^{(M-m)/2} \left| W_{\ell}(z)(\y;\x) \right| \leq C_{\ell,M} \zeta^a \mathrm{e}^{-\frac{\delta_0}{\zeta}\abs{\x-\y}}.
    \end{equation*}
    Furthermore, by using the simple inequality $(x_2-y_2)^{m}  e^{-\frac{\delta}{\zeta}\abs{\x-\y}} \leq C_{m,\delta} \zeta^{m}$ we obtain that 
    \begin{equation*}
        \abs{K_{M,m}(\xbf,\xbf)}\leq C_{\ell,M,m,\delta_0}\zeta^{a+m}\int_E \abs{\HER(z)(\xbf,\ybf)}\e^{-\frac{\delta_0}{2\zeta}\abs{\xbf-\ybf}} \dd \ybf,
    \end{equation*}
    for some constant $C_{\ell,M,m,\delta_0}$.
    To show that this integral is bounded we apply again the first resolvent identity \eqref{eq:firstresolvent} repeatedly. Choose $\lambda$ and $n$ are as before i.e.\ such that $\im[\HE,X_1] ( \HE+ \lambda)^{-n} ( \HE -z)^{-1}$ has a continuous kernel. Moreover, by an estimate similar to \eqref{eq:auxHS}, and the fact that for every bounded operator $A$ and Hilbert-Schmidt operator $B$ we have that $\norm{AB}_2 \leq \norm{A}\norm{B}_2$, where $\norm{\cdot}_2$ denotes the Hilbert-Schmidt norm, it follows that 
    \begin{equation*}
        \abs{\im[\HE,X_1] ( \HE+ \lambda)^{-n} ( \HE -z)^{-1}(\xbf,\ybf)}\leq C_{\lambda,n}\abs{\mathrm{Im}(z)^{-1}},
    \end{equation*}
    for some constant $C_{\lambda,n}$. Furthermore, by the estimates in \eqref{eq:resolvent-estimate1} and \eqref{eq:resolvent-estimate2} and the fact that $\im [\HE,X_1]=2\overline{P}_1$ it follows that $\im[\HE,X_1] ( \HE+ \lambda)^{-k}(\xbf,\cdot)\in L^1(\R^2)$ for $k\geq 1$. Now using the first resolvent identity \eqref{eq:firstresolvent} as in \eqref{eq:GreenthmInfty} gives
    \begin{equation*}
        \abs{K_{M,m}(\xbf,\xbf)}\leq C\zeta^{a+m+n+2}.
    \end{equation*}
    for some constant $C$ depending on $\ell,M,m,\delta_0,n$.
    Choosing $N$ sufficiently large, the property \ref{prop3HS} of the almost analytic extension $F_N$ immediately gives the desired result.
    \hfill \qedsymbol

\appendix 
\section{Manipulation of distributions} \label{s:dist}
In this section we collect some useful results regarding formal manipulation of distributions with specific structure. Even though these results can be seen as a standard consequence of the general theory of distributions \cite{Ho1}, we show here a proof in our setting for the convenience of the reader.

We denote by $\slowlyinc{\R^2}$ the set of slowly increasing functions on $\R^2$, i.e.\ functions $f\in \Cinf{\R^2}$ such that for all $\alpha\in \N_0^2$ there exist $c$ and $k$ (depending on $\alpha$ and $f$) such that
\begin{equation*}
    \abs{\partial^\alpha f(\xbf)}\leq c\jnorm{\xbf}^k,
\end{equation*}
for all $\xbf\in \R^2$.

Given any $\xbf\in \R^2$ and $\Lambda\in \tempdist{\R^2}$, we define $\Lambda_\xbf$ to be the temperate distribution given by 
\begin{equation}\label{eq:Lambdax}
    \pair{\Lambda_\xbf}{f}\coloneqq\pair{\Lambda}{f(\xbf-\cdot)}=(\Lambda*f)(\xbf),
\end{equation}
for any $f\in \Schwartsfunct{\R^2}$. 

 For any smooth function $\psi\in \Cinf{\R^2\times \R^2}$ we define, for any $\xbf\in \R^2$ and $\alpha\in \N_0^2$, the function $\morph{\psi_\xbf^\alpha}{\R^2}{\C}$ by 
\begin{equation*}
    \psi_\xbf^\alpha(\ybf)\coloneqq\partial_1^\alpha \psi(\xbf,\ybf).
\end{equation*} 
When $\abs{\alpha}=0$ we use the shorthand notation $\psi_\xbf$ instead of $\psi_\xbf^\alpha$.
\begin{lemma}\label{lem:diff1}
    Let $\psi\in \Cinf{\R^2\times\R^2}$ be a function such that for every $\alpha,\beta\in \N_0^2$ and $n\in \N$, there exist constants $c$ and $k$ (depending on $\alpha$, $\beta$, and $n$) such that 
    \begin{equation*}
        \jnorm{\ybf}^n \abs{\partial_1^\alpha\partial_2^\beta \psi(\xbf,\ybf)}\leq c \jnorm{\xbf}^k,
    \end{equation*}
    for all $\xbf,\ybf\in \R^2$.
    Then, for any $\Lambda\in \tempdist{\R^2}$ the function $\xbf\mapsto \pair{\Lambda}{\psi_\xbf}$ is smooth and
    \begin{equation}\label{eq:diffdist}
        \partial^\alpha\pair{\Lambda}{\psi_\xbf}=\pair{\Lambda}{\psi_\xbf^\alpha},
    \end{equation}
    for all $\alpha\in \N_0^2$.
\end{lemma}
\begin{proof}
    Since $\Lambda$ is a temperate distribution, there exist constants $C$ and $m$ such that 
    \begin{equation}\label{eq:tempdist}
        \abs{\pair{\Lambda}{f}}\leq C\sup\left\{ \jnorm{\ybf}^n \abs{\partial^\beta f(\ybf)} \mid \ybf\in \R^2,\abs{\beta}\leq n \right\},
    \end{equation}
    for all $f\in \Schwartsfunct{\R^2}$. Let $\xbf\in \R^2$, $h\in \R$, $\alpha\in \N_0^2$, and $i\in \{1,2\}$ be arbitrary. In the following $\ebf_i$ denotes the $i-$th standard vector in $\R^2$. We only prove \eqref{eq:diffdist} for $\alpha=\ebf_i$ since the general case follows by simply repeating the same argument.
    
    First notice that the function 
    \begin{equation*}
        f_{\xbf,\alpha,h}(\ybf)\coloneqq \int_0^1 \psi_{\xbf+sh\ebf_i}^\alpha(\ybf)-\psi_{\xbf}^\alpha(\ybf) \dd s
    \end{equation*}
    is in $\Schwartsfunct{\R^2}$. Indeed, for all $\beta\in \N_0^2$, $n\in \N$, and $s\in [0,1]$ the mean value theorem implies the existence of $t_s\in (0,1)$ such that
    \begin{equation}
	\begin{aligned}\label{eq:phixah}
        \jnorm{\ybf}^n\abs{\partial^\beta f_{\xbf,\alpha,h}(\ybf)}&\leq\abs{h} \int_0^1 s\jnorm{\ybf}^n\abs{\partial_{1}^{\alpha+\ebf_i}\partial^\beta_2\psi(\xbf+t_0 sh\ebf_i,\ybf)} \dd s\\
        &\leq c\abs{h}(2h^2+1)^{k/2} \jnorm{\xbf}^k,
    \end{aligned}
\end{equation}
    with $c$ and $k$ depending on $n$, $\alpha$, and $\beta$. By the fundamental theorem of calculus     
    \begin{align*}
        \abslr{h^{-1}\left(\pair{\Lambda}{\psi_{\xbf+h\ebf_i}}-\pair{\Lambda}{\psi_\xbf}\right)- \pair{\Lambda}{\psi_\xbf^\alpha}}=\abs{\pair{\Lambda}{f_{\xbf,\alpha,h}}},
    \end{align*}
    for any $h\neq 0$. By applying \eqref{eq:tempdist} and \eqref{eq:phixah} to the left hand side and taking $h\to 0$ we obtain \eqref{eq:diffdist}.
\end{proof}

\begin{lemma}\label{lem:diff2}
    Suppose that $\Lambda\in \tempdist{\R^2}$ and define
    \begin{equation*}
        \psi(\xbf,\ybf)\coloneqq g(\xbf,\xbf-\ybf)f(\xbf-\ybf),
    \end{equation*}
    with $g\in \slowlyinc{\R^2\times \R^2}$, $f\in \Schwartsfunct{\R^2}$. Then, for any $\alpha\in \N_0^2$ and $\xbf\in \R^2$,
    \begin{align}\label{eq:diff2}
        \partial^\alpha\pair{\Lambda}{\psi_\xbf} =\sum_{\beta\leq \alpha} \binom{\alpha}{\beta} \pair{ g_{\xbf}^\beta (\partial^{\alpha-\beta}\Lambda)_\xbf}{f}.
    \end{align}
\end{lemma}
\begin{proof}
    A calculation using the Leibniz rule and the chain rule shows that $\psi$ satisfies the assumptions of Lemma~\ref{lem:diff1}, and thus $\partial^\alpha\pair{\Lambda}{\psi_\xbf}=\pair{\Lambda}{\psi_\xbf^\alpha}$ for all $\xbf\in \R^2$. From this we proceed by induction on $\abs{\alpha}$. Let $\abs{\alpha}=1$. Then 
    \begin{align*}
        \psi_\xbf^\alpha(\ybf)=\partial_{1}^\alpha\psi(\xbf,\ybf)=g_\xbf^\alpha(\xbf-\ybf)f(\xbf-\ybf)-\partial^\alpha \psi_\xbf(\ybf),
    \end{align*}
    from which \eqref{eq:diff2} follows immediately (recall \eqref{eq:Lambdax}).

    For the inductive step, let $k\geq 1$ and suppose that \eqref{eq:diff2} holds when $\abs{\alpha}=k$. Let $\alpha$ be a multi-index with $\abs{\alpha}=k+1$ and write $\alpha=\alpha_0+\ebf_i$ for some $i\in \{1,2\}$. By Lemma~\ref{lem:diff1} and the induction hypothesis
    \begin{equation*}
        \pair{\Lambda}{\psi_\xbf^\alpha}=\partial^{\ebf_i}\sum_{\beta\leq \alpha_0} \binom{\alpha_0}{\beta} \pair{g_{\xbf}^\beta(\partial^{\alpha_0-\beta}\Lambda)_\xbf}{f}.
    \end{equation*}
    Applying the basis step of the induction for each term in the sum but with $\Lambda$ replaced with $\partial^{\alpha_0-\beta}\Lambda$ and $\psi$ replaced with
    $  \R^2\times\R^2\ni (\xbf,\ybf)\mapsto g_{\xbf}^\beta(\xbf-\ybf) f(\xbf-\ybf) $,
    completes the proof.
    \end{proof}

\section{Integral operators}\label{s:intop}
A central point in our proofs is to show that certain operators are integral operators and then manipulate their integral kernels. Apart from the well-known Schur test (see e.g.\ \cite[Lemma 18.1.12]{Ho3}) the following lemma is useful (see e.g.\ \cite[Lemma A.1.2]{Simon}).
\begin{lemma}\label{lem:L2toLinf}
    Let $\Omega\subset \R^2$ be an open set and let $p\in [1,\infty)$. Further let $q$ satisfy $p^{-1}+q^{-1}=1$, and let $T$ be a bounded operator on $L^p(\Omega)$. Then $T$ is a bounded operator from $L^p(\Omega)$ to $L^\infty(\Omega)$ if and only if $T$ is an integral operator with a kernel satisfying 
    \begin{equation}\label{eq:kernel_exists}
        \sup_{\xbf\in\Omega}\Bigl(\int_\Omega \abs{T(\xbf,\ybf)}^q \dd \ybf\Bigr)^{1/q}<\infty.
    \end{equation}
    Furthermore, $\norm{T}_{L^2(\Omega)\to L^\infty(\Omega)}$ equals the left hand side of \eqref{eq:kernel_exists}.
\end{lemma}

\subsection{Properties of \texorpdfstring{$S_b(-\lambda)$}{Sb(-lambda)} and \texorpdfstring{$T_b(-\lambda)$}{Tb(-lambda)}}\label{s:M}
In this section we derive properties of the operators $S_b(-\lambda)$ and $T_b(-\lambda)$ by considering their integral kernels. These results are all based on the fact that 
\begin{equation}\label{eq:K_0def}
    \FT^{-1}\Big(\frac{1}{\abs{\xi}^2+\lambda}\Big)(\xbf)=\frac{1}{2\pi}K_0(\sqrt{\lambda}\abs{\xbf}),
\end{equation}
where $K_0$ is the Macdonald function (modified Bessel function of the second kind). Recall that the Macdonald functions $K_\nu$, $\nu\in \Z$, are smooth positive functions on $(0,\infty)$ such that \cite{abramowitz1968handbook}
\begin{align}
    K_0'(x)&=-K_1(x), \quad \textup{and }\quad K_1'(x)=K_0(x)-\frac{K_1(x)}{x}\, ,\label{eq:K1}\\
    K_\nu(x)&\sim\sqrt{\frac{\pi }{2x }}\e^{-x}\quad \textup{as }x\to\infty \, , \label{eq:K2}\\
    K_0(x)&\sim-\log(x) \quad \textup{as }x\to 0\label{eq:K3}\, ,\\
    K_\nu(x)&\sim x^{-\nu}2^{\nu-1}\Gamma(\nu)\quad \textup{as }x\to0, \nu\neq 0\label{eq:K4},
\end{align}
where $f\sim g$ as $x\to a$ if and only if 
$\lim_{x\to a} \tfrac{f(x)}{g(x)}=1$.
Note that both $K_0(\abs{\cdot})$ and $K_1(\abs{\cdot})$ belong to $L^1(\R^2)$. By using the properties \eqref{eq:K2} and \eqref{eq:K3} in \eqref{eq:Slambda1}, it follows that there exists $C>0$ such that for all $\lambda>0$,

    \begin{equation}\label{eq:S-estimate}
        \abs{S_b(-\lambda)(\xbf,\ybf)}\leq C(1+\abs{\ln(\sqrt{\lambda}\abs{\xbf-\ybf})})\e^{-\sqrt{\lambda}\abs{\xbf-\ybf}}.
    \end{equation}
Combining this with Schur's test and Lemma~\ref{lem:L2toLinf}, immediately gives the following result. 
\begin{lemma}\label{lem:S}
    The operator $S_b(-\lambda)$ is a bounded self-adjoint operator on $L^2(E)$ and also bounded as an operator from $L^2(E)$ to $L^\infty(E)$.
\end{lemma}
We are also interested in the regularity of $S_b(-\lambda)f$ when $f$ is a test function. Recall that $\M$ denotes the set smooth functions on $E$ with exponentially decaying derivatives of all orders such that $f(x_1,0)\equiv0$ for all $x_1\in \R$ (cf.\ Lemma~\ref{lem:MinDHb}).
\begin{lemma}\label{lem:SM}
    The operator $S_b(-\lambda)$ maps $\Ccinf{E}$ into $\M$.
\end{lemma}
\begin{proof}
    As in the proof of Proposition~\ref{prop:1}, we will argue using distributional derivatives (cf.\ Appendix~\ref{s:dist}) in order to avoid problems with the singularities of $S_b(-\lambda)$ along the diagonal $\xbf=\ybf$. Following the strategy in the proof of Proposition~\ref{prop:1}, let $f\in \Ccinf{E}$ and write 
 \begin{equation}\label{eq:Slambda}
	   \begin{aligned}
        \left(S_b(-\lambda)f\right)(\xbf)=\pair{\mathfrak{f}_{1,0}}{\e^{\im b\phi(\xbf,\xbf-\cdot)}f(\xbf-\cdot)}-\pair{\mathfrak{f}_{1,0}}{\e^{\im b\phi^*(\xbf,\xbf-\cdot)}f^*(\xbf-\cdot)},
    \end{aligned}
\end{equation}
    for all $\xbf\in E$, where $\mathfrak{f}_{1,0}$ is the distribution given by $(2\pi)^{-1}K_0(\sqrt{\lambda}\abs{\cdot})$ (cf.\ the notation introduced in \eqref{eq:fna}). Note that we have used the compact support of $f$ in $E$ to extend the domain of integration to all of $\R^2$. 
    
    By Lemma~\ref{lem:diff1}, it follows directly that $S_b(-\lambda)f$ is smooth. Furthermore, from the exponential decay of $K_1(\sqrt{\lambda}\abs{\cdot})$ (cf.\ \eqref{eq:K2}), the compact support of $f$, and the inequality (valid for all $\xbf,\wbf\in E$ and $c>0$)
    \begin{align*}
        \e^{c\abs{\xbf}}\leq\e^{c\abs{\xbf-\wbf}}\e^{c\abs{\wbf}}=\e^{c\abs{(\xbf-\wbf)^*}}\e^{c\abs{\wbf}},
    \end{align*}
    we obtain \eqref{eq:Mdef}.
    
    It remains to show that $(S_b(-\lambda)f)(x_1,0)=0$ for any $x_1\in \R$. To prove this, one simply uses that $K_0(\sqrt{\lambda}\abs{\cdot})$ is a radial function together with the change of variable $\wbf\mapsto \wbf^*$ in the first term in \eqref{eq:Slambda}.
\end{proof}
Consider now $T_b(-\lambda)$, by using \eqref{eq:K1}--\eqref{eq:K4} in \eqref{eq:Tlambda} we get that there exists $C>0$ such that for all $\lambda>0$,
\begin{equation}\label{eq:T-estimate}
    \abs{T_b(-\lambda)(\xbf,\ybf)}\leq C(b^2+b)\frac{\lambda+1}{\lambda}\e^{-\frac{\sqrt{\lambda}}{2}\abs{\xbf-\ybf}}.
\end{equation}
Combining \eqref{eq:T-estimate} with Schur's test immediately gives the following result.
\begin{lemma}\label{lem:T}
    The operator $T_b(-\lambda)$ is bounded on $L^2(E)$, with $\norm{T_b(-\lambda)}\leq C\abs{b^2+b}\frac{1+\lambda}{\lambda^2}$ for some constant $C$ which is independent of $\lambda>0$ and $b$.
\end{lemma}

\subsection{Resolvent estimates}\label{sec:estimates}
{
In this section we collect useful estimates of various integral kernels used in the proofs.

Combining \eqref{eq:HE_resolvent} with the estimates \eqref{eq:S-estimate} and \eqref{eq:T-estimate} gives that for $\lambda>0$ sufficiently large, there exist $C,\delta>0$ such that
\begin{equation}\label{eq:Rb-estimate}
    \abs{(\HEb+\lambda)^{-1}(\xbf,\ybf)}\leq C(1+\abs{\ln(\sqrt{\lambda}\abs{\xbf-\ybf})})\e^{-\delta\sqrt{\lambda}\abs{\xbf-\ybf}}.
\end{equation}
By arguing as in the proof of Proposition~\ref{prop:1}, it follows that $\overline{P_j}S_b(-\lambda)$ is an integral operator with kernel satisfying
    \begin{equation}\label{eq:pjs}
        \begin{aligned}
            \overline{P_j}S_b(-\lambda)(\xbf,\ybf)=\e^{\im b\phi(\xbf,\ybf)}[&-b(A_t(\xbf-\ybf))_j-\mathcal{A}_j(\xbf)\\
            &+\im(\partial_j\mathcal{A}_j(\xbf))-\im \nabla_1]R_E(\xbf,\ybf).
        \end{aligned}
    \end{equation}
Since $\overline{W}=-2\mathcal{A}\cdot P-\im(\nabla\cdot\mathcal{A})+\mathcal{A}^2+V$ it is not difficult to obtain that for $\lambda>0$ sufficiently large, exist $C,\delta>0$ such that for any $U\in\{ \overline{P_j}S_b(-\lambda),\overline{W}S_b(-\lambda), \overline{W}(\HEb+\lambda)^{-1}\}$ we have the estimate
\begin{equation}\label{eq:1divx-estimate}
    \abs{U(\xbf,\ybf)}\leq C(1+\abs{\xbf-\ybf}^{-1})\e^{-\delta\sqrt{\lambda}\abs{\xbf-\ybf}},
\end{equation}
where $C$ only depends on the potential $V$ through $\norm{V}_{\infty}$ when $U$ is either $ \overline{W}S_b(-\lambda)$ or $\overline{W}(\HEb+\lambda)^{-1}$. 

Therefore, putting together the estimates in \eqref{eq:Rb-estimate} and \eqref{eq:1divx-estimate} with the proof of Theorem~\ref{thm:1}\ref{thm:1Part2}, we get that for $\lambda>0$ sufficiently large, there exist $C,\delta>0$ such that
\begin{equation}\label{eq:resolvent-estimate1}
    \abs{(\HE+\lambda)^{-1}(\xbf,\ybf)}\leq C(1+\abs{\ln(\sqrt{\lambda}\abs{\xbf-\ybf})})\e^{-\delta\sqrt{\lambda}\abs{\xbf-\ybf}}
\end{equation}
and
\begin{equation}\label{eq:resolvent-estimate2}
    \abs{U(\HE+\lambda)^{-1}(\xbf,\ybf)}\leq  C(1+\abs{\xbf-\ybf}^{-1})\e^{-\delta\sqrt{\lambda}\abs{\xbf-\ybf}}.
\end{equation}
where $U\in \{\overline{W}, \overline{P}_j\}$.
}

Then, in order to extend \eqref{eq:resolvent-estimate1} and \eqref{eq:resolvent-estimate2} from $\lambda \in \R$ to $z$ in the complex plane, we can use the results showed in \cite[Proposition B.1, Corollary B.4]{CorneanFournaisFrankHelffer} and extend them {\it mutatis mutandis} to our setting since the potentials we consider are of class $BC^\infty(\R^2,\R)$. Therefore, we get that there exist $C,\delta>0$ such that 
\begin{equation}\label{eq:HEz}
    \abs{(\HE-z)^{-1}(\xbf,\ybf)}\leq C\zeta^2 (1+\abslr{\ln(\abs{\xbf-\ybf})})e^{-\delta\zeta^{-1}\abs{\xbf-\ybf}},
\end{equation}
for all $z\in \C$ with $\abs{\mathrm{Im}(z)}\in (0,1]$ and where $\zeta=\jnorm{\mathrm{Re(z)}}\abs{\mathrm{Im}(z)}^{-1}$.

By considering the first resolvent identity coupled with the estimates \eqref{eq:HEz}, \eqref{eq:resolvent-estimate2} and the theory of polar operators \cite[Lemma 2, p 213]{Vladimirov} (see also \cite{Pankrashkin}), we get that there exist $C',\delta'>0$ such that
\begin{equation}\label{eq:DHEz}
    \abs{\overline{P}_j(\HE-z)^{-1}(\xbf,\ybf)}\leq C\zeta^3 (1+\abs{\xbf-\ybf}^{-1})e^{-\delta\zeta^{-1}\abs{\xbf-\ybf}},
\end{equation}
for all $z\in \C$ with $\abs{\mathrm{Im}(z)}\in (0,1]$. 

\subsection{Regularity of products of kernel operators}\label{s:kernel}
The aim of this section is to describe the regularity of products of the operators $S_b(-\lambda)$, $T_b(-\lambda)$, $\overline{P_j}S_b(-\lambda)$, and $\overline{W}S_b(-\lambda)$ for $\lambda>0$. The main hurdle when considering the regularity of products of such operators is that their integral kernels do not necessarily remain integrable when differentiated. The key observation for circumventing this problem is that all four operators can be written as a sum of integral operators whose integral kernels have a specific structure that allows us to use the regularity properties of the Fourier transform.

Define for any $n\in \N$ and $\alpha\in \N_0^2$, 
\begin{equation}\label{eq:fna}
    \mathfrak{f}_{n,\alpha}\coloneqq\FT^{-1}\Big( \frac{\xi^\alpha}{(\abs{\xi}^2+\lambda)^n} \Big),
\end{equation}
where the Fourier transform should be understood in the distributional sense. By elementary properties of the Fourier transform it follows from \eqref{eq:Slambda1}, \eqref{eq:Tlambda}, \eqref{eq:pjs}, and \eqref{eq:W} that the operators $S_b(-\lambda)$, $T_b(-\lambda)$, $\overline{P_j}S_b(-\lambda)$, and $\overline{W}S_b(-\lambda)$ can be written as sums of integral operators with integral kernels taking the forms
\begin{numcases}{K(\xbf,\ybf)=}
        \e^{\im b\phi(\xbf,\ybf)}a(\xbf)\mathfrak{f}_{n,\alpha}(\xbf-\ybf),\label{eq:kernels1}\\
        \e^{\im b\phi(\xbf,\ybf)}a(\xbf)h(\xbf-\ybf)\mathfrak{f}_{n,\alpha}(\xbf-\ybf^*),\label{eq:kernels2}
\end{numcases}
where $a\in \BCinf{\R^2}$, $h\in \slowlyinc{\R^2}$ and $\abs{\alpha}\leq n$. In the following we will always assume that $a$, $h$, $\alpha$ and $n$ satisfy these assumptions when considering kernels of the form \eqref{eq:kernels1} or \eqref{eq:kernels2}. The main advantage in writing the kernels using the inverse Fourier transform instead of the functions $K_\nu$ is that it allows to easily apply the convolution theorem (cf. \cite[Theorem 30.4]{Treves} and \eqref{eq:convolution} below) when dealing with products of kernels like the one in \eqref{eq:kernels1}.

Before going into the mathematical details of the proofs, let us state the main result here. 
\begin{lemma}\label{lem:stws}
    Let $g\in \Ccinf{E}$. Furthermore, let $\lambda>0$ and $b\in \R$ be arbitrary and let $Y_1=S_b(-\lambda)$, $Y_2=T_b(-\lambda)$, $Y_3=\overline{P_j}S_b(-\lambda)$, $Y_4=\overline{W}S_b(-\lambda)$. Then, for any $m\in \N$ there exist $N\in \N$ and $\epsilon>0$ (depending on $\lambda$ and $m$) such that for all $k_1,\dots,k_N\in \{1,2,3,4\}$, the operators $gY_{k_1}\dots Y_{k_N} \e^{\epsilon \abs{X}}$ map $L^2(E)$ continuously into $C^{m}(E)$, i.e.\ there exists $C>0$ such that for every $f\in L^2(E)$,
\begin{equation*}
    \norm{gY_{k_1}\dots Y_{k_N} \e^{\epsilon \abs{X}}f}_{C^m(E)}\leq C\norm{f}_{L^2(E)}.
\end{equation*}
\end{lemma}
\begin{proof}
    As discussed above, the product $Y_{k_1}\dots Y_{k_N}$ can be written as a sum of integral operators with kernels of the form
\begin{equation}\label{eq:K}
    K(\xbf,\xbf')=\int_E\cdots\int_E K_1(\xbf,\ybf_1)\dots K_{N}(\ybf_{N-1},\xbf')\dd \ybf_1\dots \dd\ybf_{N-1},
\end{equation}
where each $K_j$ is either given by \eqref{eq:kernels1} or \eqref{eq:kernels2}. For every such $K$, if $N\geq m+3$, Lemma~\ref{lem:K1}, Lemma~\ref{lem:K2}, and Lemma~\ref{lem:K3} below give, for every $\alpha\leq m$ the existence of $C,\epsilon>0$ such that  
\begin{equation}
\label{eq:aux21}
   \abs{ \partial_1^\alpha( g(\xbf) K(\xbf,\xbf'))}\leq Ce^{-\epsilon\abs{\xbf'}},
\end{equation}
for all $\xbf,\xbf'\in E$. Note that by the proofs of Lemma~\ref{lem:K2} and Lemma~\ref{lem:K3}, the localization obtained by multiplying with $g$ is necessary. Note also that the condition $N\geq m+3$ is due to Lemma~\ref{lem:K2}. From inequality \eqref{eq:aux21} and Lebesgue's dominated convergence theorem the result follows. 
\end{proof}

The rest of the section is devoted to proving the three results used in the proof of Lemma~\ref{lem:stws}, namely Lemma~\ref{lem:K1}, Lemma~\ref{lem:K2}, and Lemma~\ref{lem:K3} below. These lemmas regard the regularity of product of integral kernels of the form \eqref{eq:K} with $K_j$ given as in either \eqref{eq:kernels1} or \eqref{eq:kernels2}. Specifically, Lemma~\ref{lem:K1} covers the cases when $K_1$ is of the form \eqref{eq:kernels2}, Lemma~\ref{lem:K2} covers the case when no $K_j$ is of the form \eqref{eq:kernels2}, and Lemma~\ref{lem:K3} the case when for some $j\geq 2$, $K_j$ is of the form \eqref{eq:kernels2}. 

By the convolution theorem (cf. \cite[Theorem 30.4]{Treves}) it follows that 
\begin{equation}\label{eq:convolution}
    \mathfrak{f}_{n+m,\alpha+\beta}=\FT^{-1}\Big(\frac{\xi^\alpha}{(\abs{\xi}^2+\lambda)^n}\frac{\xi^\beta}{(\abs{\xi}^2+\lambda)^m}\Big)= \mathfrak{f}_{n,\alpha}* \mathfrak{f}_{m,\beta},
\end{equation}
for all $\alpha,\beta\in \N_0^2$ and $n,m\in \N$. If $\abs{\alpha}\leq n$, then $f_{n,\alpha}$ may be written as a finite convolution of $f_{1,\beta}$ where $\abs{\beta}\leq 1$ and, from the properties of $K_0$ and $K_1$, it follows that $\mathfrak{f}_{n,\alpha}\in \Cinf{\R^2\setminus\{0\}}\cap L^1(\R^2)$ and that for any $c>0$ and $\beta\in \N_0^2$ there exist $C,\delta>0$ such that 
\begin{equation}\label{eq:ffrak}
    \abs{\partial^\beta\mathfrak{f}_{n,\alpha}(\xbf)}\leq C\e^{-\delta\abs{\xbf}},
\end{equation}
for all $\abs{\xbf}\geq c$. Hence, by using the fact that $
    \abs{\xbf-\ybf}\leq \abs{\xbf-\ybf^*}$, for all $\xbf,\ybf\in E$, it follows that, if $K_j$ is of either the form \eqref{eq:kernels1} or \eqref{eq:kernels2}, then there exists some $\delta_0>0$ such that 
    \begin{equation}\label{eq:qintegral}
        \max\bigg\{ \sup_{\xbf\in E} \int_{E}\e^{\delta\abs{\xbf-\ybf}}\abs{K_j(\xbf,\ybf)}\dd \ybf , \sup_{\ybf\in E} \int_{E}\e^{\delta\abs{\xbf-\ybf}}\abs{K_j(\xbf,\ybf)}\dd \xbf\bigg\} <\infty
    \end{equation}
for all $\delta\in [0,\delta_0)$.

Below we prove the first case where $K_1$ is of the form \eqref{eq:kernels2}. This is the easiest case since the presence of $\ybf^*$ in the formula for $K_1(\xbf,\ybf)$ and a localizing function $g\in \Ccinf{E}$ allow us to simply use Lebesgue's dominated convergence theorem when differentiating $K$. This result will be applied several times when dealing with the other two cases.
\begin{lemma}\label{lem:K1}
    Let $K$ be as in \eqref{eq:K} and assume that $K_1$ is of the form \eqref{eq:kernels2}, i.e.\ $K_1(\xbf,\ybf)=\e^{\im b\phi(\xbf,\ybf)}a_1(\xbf)h_1(\xbf-\ybf)\mathfrak{f}_{n_1,\alpha_1}(\xbf-\ybf^*)$. Then, for any $g\in \Ccinf{E}$ and $\alpha\in\N_0^2$, there exist $C,\delta>0$ such that 
    \begin{equation}\label{eq:qestimate}
        \abs{\partial_1^\alpha g(\xbf)K(\xbf,\xbf')}\leq C\e^{-\delta\abs{\xbf'}}
    \end{equation}
    for all $\xbf\in E,\xbf'\in E$. 
\end{lemma}
\begin{proof}
    Let $g\in \Ccinf{E}$ be arbitrary. Note first that all the partial derivatives of $\e^{\im b\phi(\xbf,\ybf)}h_1(\xbf-\ybf)a_1(\xbf)$ are bounded on $E\times E$ by some polynomial in $\abs{\xbf-\ybf^*}$. Hence, by using the exponential decay given by \eqref{eq:ffrak} and the fact that the singularity of $K_1$ plays no role since $\dist(\Supp g, \partial E)>c>0$ for some $c$, for all $\alpha\in \N_0^2$ there exist $C,\delta>0$ such that
    \begin{equation}
    \label{eq:aux212}
        \abs{\partial_1^\alpha g(\xbf)K_1(\xbf,\ybf)}\leq\chi_{\Supp g}(\xbf)C\e^{-\delta\abs{\xbf-\ybf^*}},
    \end{equation}
    where $\chi_{\Supp g}$ denotes the indicator function of the support of $g$. Then, triangle inequality gives
    \begin{equation}\label{eq:ineq1}
        \e^{-\delta\abs{\xbf-\ybf^*}}\leq \e^{-\delta\abs{\xbf-\ybf}}\leq \e^{\delta\abs{\xbf}}\e^{-\delta\abs{\ybf}},
        \end{equation}
    From \eqref{eq:ineq1}, the above estimate, \eqref{eq:qintegral}, and Lebesgue's dominated convergence theorem the result follows.
\end{proof}
When $K_1(\xbf,\ybf)=\e^{\im b\phi(\xbf,\ybf)}a_1(\xbf)\mathfrak{f}_{n_1,\alpha_1}(\xbf-\ybf)$ there are two basic cases to consider. Either every $K_j$ is of the same form of $K_1$, or at least one $K_j$ is of the form \eqref{eq:kernels2}. In both cases, the following lemma is used to handle products of two kernels of the form $\e^{\im b\phi(\xbf,\ybf)}a(\xbf)\mathfrak{f}_{n,\alpha}(\xbf-\ybf)$ by using the convolution property of the Fourier transform and the regularity properties of $\mathfrak{f}_{\eta,\alpha}$.

\begin{lemma}\label{lem:Fourier_decay2}
    Let $K_1,K_2$ be of the form in \eqref{eq:kernels1}. Then, for any $M\in \N$, there exists a function $G(\cdot,\cdot)\colon \R^2 \times \R^2 \to \C$ such that

    \begin{enumerate}[label=(\roman*)]
        \item for all $\xbf'\in \R^2$, $G(\cdot,\xbf')\in C^{M}(\R^2)$.
        \item For all $\alpha\in\N_0^2$ with $\abs{\alpha}\leq M$ there exist $C,\delta>0$ such that
        \begin{equation}\label{eq:Gexpdecay}
            \abs{\partial_1^{\alpha}G(\xbf,\xbf')}\leq C\e^{-\delta\abs{\xbf-\xbf'}}, \qquad \forall \, \xbf,\xbf'\in \R^2\, .
        \end{equation}
        \item \label{lem:Fourier_decay2_p3} The product kernel $K_{\R^2}(\xbf,\xbf'):=\int_{\R^2} K_1(\xbf,\ybf)K_2(\ybf,\xbf')\dd \ybf$ can be written as 
        \begin{equation}\label{eq:almostconvolution}
            K_{\R^2}(\xbf,\xbf')=\e^{\im b\phi(\xbf,\xbf')}G(\xbf,\xbf')+\e^{\im b\phi(\xbf,\xbf')}\sum_{n,\eta} a_{n,\eta}(\xbf) \mathfrak{f}_{n,\eta} (\xbf-\xbf'),
        \end{equation}
        where the sum is taken over all pairs $(n,\eta)\in \N\times\N_0^{2}$ such that $\abs{\eta}\leq \abs{\alpha_1}+\abs{\alpha_2}+2M+2$ and
        \begin{equation*}
            \abs{\eta}-2n=\abs{\alpha_1}+\abs{\alpha_2}-2n_1-2n_2-k
        \end{equation*}
        for some $k\in \{0,\dots,2M+2\}$. Furthermore, $a_{n,\eta}\in \BCinf{\R^2}$.
    \end{enumerate}
    
\end{lemma}
\begin{proof}
   Let us recall the notation for $K_j$, i.e. $
    K_j(\xbf,\ybf)=:\e^{\im b\phi(\xbf,\ybf)}a_j(\xbf)\mathfrak{f}_{n_j,\alpha_j}(\xbf-\ybf)$,  for $j\in \{1,2\}$. By adding and subtracting $a_2(\xbf)$ and by using the composition rule of the Peierls phase \eqref{eq:phis}, we may rewrite $K_{\R^2}$ as
    \begin{align*}
        K_{\R^2}(\xbf,\xbf')&=\e^{\im b\phi(\xbf,\ybf)}a_1(\xbf)\bigg[a_2(\xbf)\!\!\int_{\R^2}\!\!\!\! \e^{\im b\phi_s(\xbf-\ybf,\ybf-\xbf')}\mathfrak{f}_{n_1,\alpha_1}(\xbf-\ybf)\mathfrak{f}_{n_2,\alpha_2}(\ybf-\xbf')\dd \ybf\\
        &\quad -\int_{\R^2}\!\!\!\! \e^{\im b\phi_s(\xbf-\ybf,\ybf-\xbf')}\mathfrak{f}_{n_1,\alpha_1}(\xbf-\ybf)(a_2(\xbf)-a_2(\ybf))\mathfrak{f}_{n_2,\alpha_2}(\ybf-\xbf')\dd \ybf\bigg]\\
        &\eqqcolon \widetilde{K}(\xbf,\xbf')-\widehat{K}(\xbf,\xbf').
    \end{align*}
    Let us first consider $\tilde{K}$.
    The sum in \eqref{eq:almostconvolution} is a consequence of Taylor's theorem and \eqref{eq:convolution}.
    Specifically, by expanding $\e^{\im b(\cdot)\phi_s(\xbf-\ybf,\ybf-\xbf') }\colon \R \to \C$ up to the order $M+1$ and using the integral form of the remainder we get
    \begin{equation*}
        \widetilde{K}(\xbf,\xbf')=\e^{\im b\phi(\xbf,\xbf') }a_1(\xbf)a_2(\xbf)\Big(G(\xbf,\xbf')+ \sum_{k=0}^{M+1} G_k(\xbf,\xbf')\Big)
    \end{equation*} 
    where 
    \begin{align*}
        G(\xbf,\xbf')=\frac{1}{(M+1)!}\int_{\R^2} \bigg[&(\im b\phi_s(\xbf-\ybf,\ybf-\xbf'))^{M+2}\mathfrak{f}_{n_1,\alpha_1}(\xbf-\ybf) \mathfrak{f}_{n_2,\alpha_2}(\ybf-\xbf')\\
        &\cdot\int_0^1 \e^{\im s b\phi_s(\xbf-\ybf,\ybf-\xbf')}(1-s)^{M+1}\dd s \bigg] \dd \ybf
    \end{align*}
    and 
    \begin{equation*}
        G_k(\xbf,\xbf')=\frac{1}{k!}\int_{\R^2} (\im b\phi_s(\xbf-\ybf,\ybf-\xbf'))^{k} \mathfrak{f}_{n_1,\alpha_1}(\xbf-\ybf) \mathfrak{f}_{n_2,\alpha_2}(\ybf-\xbf') \dd \ybf.
    \end{equation*}
    Let us consider the integral kernel given by $G(\x,\x')$. By using the binomial theorem to expand $\phi_s(\xbf-\ybf,\ybf-\xbf')^{M+2}$ and then moving the terms inside the (inverse) Fourier transforms in $\mathfrak{f}_{\alpha_1,n_1}$ and $\mathfrak{f}_{\alpha_2,n_2}$ we obtain 
\begin{equation}
	    \begin{aligned}\label{eq:Grewritten}
        G(\xbf,\xbf')=\sum_{m_1,\eta_1,m_2,\eta_2}c_{m_1,\eta_1,m_2,\eta_2}\int_{\R^2}\bigg[&\mathfrak{f}_{m_1,\eta_1}(\xbf-\ybf) \mathfrak{f}_{m_2,\eta_2}(\ybf-\xbf')\\
        &\cdot\int_0^1 \frac{\e^{\im  sb\phi_s(\xbf-\ybf,\ybf-\xbf')}}{(M+1)!} (1-s)^{M+1}\dd s\bigg] \dd \ybf,
    \end{aligned}
\end{equation}
    where the sum is taken over all $(m_1,\eta_1),(m_2,\eta_2)\in\N\times \N_0^2$ with $\abs{\eta_j}\leq \abs{\alpha_j}+M+2$ and 
    $
        \abs{\eta_j}-2m_j=\abs{\alpha_j}-M-2-2n_j,
    $
    for $j\in \{1,2\}$. Combining \eqref{eq:Grewritten} with standard results for the Fourier transform and the fact that $2n_j-\abs{\alpha_j}\geq 1$ for $j\in \{1,2\}$, we obtain that $\mathfrak{f}_{m,\eta}$ is at least in $C^M(\R^2)$ for all pairs $(m,\eta)$ occurring in the sum in \eqref{eq:Grewritten}. Since also 
    \[\partial_\xbf^\alpha \e^{\im bs\phi_s(\xbf-\ybf,\ybf-\xbf')}=(\im b s)^{\abs{\alpha}} (\ybf-\xbf')^{\tilde{\alpha}}\e^{\im b s\phi_s(\xbf-\ybf,\ybf-\xbf')},\]
    where $\tilde{\alpha}=[\alpha_2,\alpha_1]$, it follows from Lebesgue's dominated convergence theorem that for all $\xbf'\in \R^2$, $G(\cdot,\xbf')\in C^{M}(\R^2)$. Then, since $a_j \in BC^{\infty}(\R^2)$, a simple application of the Leibinz rule, the triangle inequality, and \eqref{eq:ffrak} imply \eqref{eq:Gexpdecay}.

    The integral kernels associated to $G_k(\x,\x')$ can be treated in a similar way, notice that the lack of the integral remainder makes it possible to directly use \eqref{eq:convolution} to obtain the sum in \eqref{eq:almostconvolution}. 
    
    It remains to treat the integral kernel $\widehat{K}(\x,\x')$, where first we have to expand $a_2(\xbf)-a_2(\ybf)$ using Taylor's formula, and then concluding using the same strategy as before.
\end{proof}
Next we use Lemma~\ref{lem:Fourier_decay2} to treat arbitrary products of kernels of the form \eqref{eq:kernels1}. However, as Lemma~\ref{lem:Fourier_decay2}\ref{lem:Fourier_decay2_p3} concerns integration over all of $\R^2$ we will need to first extend the domain of integration from $E$ to $\R^2$ in \eqref{eq:K} by  using Lemma~\ref{lem:K1}. 

\begin{lemma}\label{lem:K2}
    Let $N\geq 2$ and let $K$ be as in \eqref{eq:K} with $K_j$ of the form in \eqref{eq:kernels1} for $j\in \{1,\dots,N\}$.

    Then, for any $M\in \N$ there exists a function $H(\xbf,\xbf')$ satisfying:
    \begin{enumerate}[label=(\roman*)]
        \item \label{lem:K2Point1} For any $g\in \Ccinf{E}$ and all $\xbf'\in E$, $g(\cdot)H(\cdot,\xbf')\in C^{M}(E)$.
        \item \label{lem:K2Point2} For any $g\in \Ccinf{E}$ and all $\alpha\in\N_0^2$ with $\abs{\alpha}\leq M$ there exist $C,\delta>0$ such that
        \begin{equation}\label{eq:Hexpdecay}
            \abs{\partial_1^{\alpha}g(\xbf)H(\xbf,\xbf')}\leq C\e^{-\delta\abs{\xbf'}},
        \end{equation}
        for all $\xbf,\xbf'\in E$.
        \item \label{lem:K2Point3}The product kernel $K(\xbf,\xbf')$ can be written as 
        \begin{equation}\label{eq:almostconvolution1}
            K(\xbf,\xbf')=H(\xbf,\xbf')+\e^{\im b\phi(\xbf,\xbf')}\sum_{n,\eta} a_{n,\eta}(\xbf) \mathfrak{f}_{n,\eta} (\xbf-\xbf'),
        \end{equation}
        where the sum is taken over all pairs $(n,\eta)\in \N\times\N_0^{2}$ such that $\abs{\eta}\leq \sum_{j=1}^N \abs{\alpha_j}+2(N-1)(M+1)$ and
        \begin{equation*}
            \abs{\eta}-2n=\sum_{j=1}^N(\abs{\alpha_j}-2n_j)-k
        \end{equation*}
        for some $k\in \{0,\dots,2(N-1)(M+1)\}$. Furthermore, $a_{n,\eta}\in \BCinf{\R^2}$.
    \end{enumerate}
    In particular, if $\tilde{N}=\sum_{j=1}^N(n_j-\abs{\alpha_j})$ and $g\in \Ccinf{E}$ then $g(\cdot)K(\cdot,\xbf')\in C^{\tilde{N}-3}(E)$ for $\xbf'\in E$ and additionally, for all $\alpha\in \N_0^2$ with $\abs{\alpha}\leq \tilde{N}-3$ there exist $C,\delta>0$ such that 
    \begin{equation}\label{eq:Kestimate}
        \abs{\partial_1^\alpha K(\xbf,\xbf')}\leq C\e^{-\delta\abs{\xbf'}}
    \end{equation}
    for all $\xbf,\xbf'\in E$.
\end{lemma}
\begin{proof}
    Let $g\in \Ccinf{E}$ be arbitrary. The first step in the proof is to extend the domain of integration for the variable $\ybf_1$ in \eqref{eq:K} from $E$ to $\R^2$. Hence, we write
    \begin{align*}
        K(\xbf,\xbf')&=\int_E\dots\int_E \int_{\R^2}K_1(\xbf,\ybf_1)K_2(\ybf_1,\ybf_2)\dots K_N(\ybf_{N-1},\xbf') \dd \ybf_1\dots \dd \ybf_{N-1}\\
        &\quad- \int_E\dots\int_E \int_{-E} K_1(\xbf,\ybf_1)\dots K_N(\ybf_{N-1},\xbf') \dd \ybf_1\dots \dd \ybf_{N-1}\\
        &\eqqcolon \tilde{K}(\xbf,\xbf')-\widehat{K}(\xbf,\xbf').
    \end{align*}
    By applying the change of variable $\ybf_1\mapsto \ybf_1^*$ in $\widehat{K}$ and taking into account that $\partial_1^\alpha\phi(\xbf,\ybf^*)$ is still polynomially bounded in $\abs{\xbf-\ybf^*}$ for $\xbf,\ybf\in E$, the arguments used in the proof of Lemma~\ref{lem:K1} show that $g\widehat{K}$ satisfies \eqref{eq:qestimate}.

    The next step is to examine $\tilde{K}$. By applying Lemma~\ref{lem:Fourier_decay2} to the integral with respect to $\ybf_1$ in $\tilde{K}$, i.e.\ $\int_{\R^2}K_1(\xbf,\ybf_1)K_2(\ybf_1,\ybf_2)\dd \ybf_1$, we can write $\tilde{K}$ as  
    \begin{align*}
        \tilde{K}(\xbf,\xbf')&=\sum_{n,\eta}a_{n,\eta}(\xbf)\int_E\dots\int_E \bigg[\e^{\im b\phi(\xbf,\ybf_2)}\mathfrak{f}_{n,\eta}(\xbf-\ybf_2)K_3(\ybf_2,\ybf_3)\\
        &\phantom{=\sum_{n,\eta}a_{n,\eta}(\xbf)\int_E\dots\int_E \bigg[}\cdots K_N(\ybf_{N-1},\xbf')\bigg] \dd \ybf_2\dots \dd \ybf_{N-1}\\
        &\quad +\int_E\dots\int_E \bigg[\e^{\im b\phi(\xbf,\ybf_2)}G(\xbf,\ybf_2) K_3(\ybf_2,\ybf_3)\\
        &\phantom{\quad +\int_E\dots\int_E \bigg[}\cdots K_N(\ybf_{N-1},\xbf')\bigg] \dd \ybf_2\dots \dd \ybf_{N-1}\\
        &\eqqcolon\sum_{n,\eta} \tilde{K}_{n,\eta}(\xbf,\xbf')+\tilde{G}(\xbf,\xbf'),
    \end{align*} 
    where $G$, $\eta$, $n$ and $a_{n,\eta}$ are as in Lemma~\ref{lem:Fourier_decay2}. In particular, we have that all pairs $(n,\eta)$ in the sum above satisfy $2n-\abs{\eta}\geq 2(n_1+n_2)-\abs{\alpha_1}+\abs{\alpha_2}$, and that $G(\cdot,\xbf')\in C^M(\R^2)$ for all $\xbf'\in \R ^2$, with derivatives that decay exponentially in $\abs{\xbf-\xbf'}$.    

    Next we argue that $\tilde{G}$ satisfies \ref{lem:K2Point1} and \ref{lem:K2Point2}. Note that, for any $\alpha$ there exists $k$ such that 
    \begin{equation}\label{eq:ineq2}
        \abs{\partial_1^\alpha \phi(\xbf,\ybf)}\leq C\jnorm{\xbf-\ybf^*}^k\leq C\jnorm{\xbf}^k\jnorm{\ybf}^k,
    \end{equation} 
    which implies that it is not possible to use Lebesgue's dominated convergence theorem when differentiating $\tilde{G}$ without also having the localization provided by $g$. However, \eqref{eq:ineq1}, \eqref{eq:ineq2}, \eqref{eq:Gexpdecay}, and \eqref{eq:qintegral} imply that $g\tilde{G}$ satisfies point \ref{lem:K2Point1} and \ref{lem:K2Point2}  above. 

    After that, for each $K_{n,\eta}$ we repeat the process of extending the domain and applying Lemma~\ref{lem:Fourier_decay2} until we obtain the formula 
    \begin{equation*}
        \tilde{K}_{n,\eta}(\xbf,\xbf')=H_{n,\eta}(\xbf,\xbf')+ \e^{ \im b \phi(\xbf,\xbf')}\sum_{n',\eta'} a^{({n,\eta})}_{n',\eta'}(\xbf) \mathfrak{f}_{n',\eta'}(\xbf-\xbf'),
    \end{equation*}
    where $H_{n,\eta}$ satisfies \ref{lem:K2Point1} and \ref{lem:K2Point2} above, and the sum is taken over all $(n',\eta')\in \N\times \N_0^2$ with $\abs{\eta'}\leq \sum_{j=1}^N \abs{\alpha_j}+2(N-1)(M+1)$ and 
    $\abs{\eta'}-2n=\sum_{j=1}^N(\abs{\alpha_j}-2n_j')-k$
    for some $k\in \{0,\dots,2(N-1)(M+1)\}$. 
\end{proof}

We are now in a position to consider the last remaining case, where $K_1$ is of the form \eqref{eq:kernels1} and $K_j$ is of the form \eqref{eq:kernels2} for some $j\geq 2$.

\begin{lemma}\label{lem:K3}
    Let $N\geq 2$ and suppose that $K$ is as in \eqref{eq:K} with $K_j$ is either of the form \eqref{eq:kernels1} or \eqref{eq:kernels2}, for $j\in \{1,\dots,N\}$. Moreover, suppose that for some $j_0\in \{1,\dots,N\}$, $K_{j_0}$ is of the form \eqref{eq:kernels2}.

    Then, for any $M\in \N$ and $g\in \Ccinf{E}$, the function $g(\cdot)K(\cdot,\xbf')\in C^{M}(E)$ for all $\xbf'\in E$. Additionally, for all $\alpha\in\N_0^2$ with $\abs{\alpha}\leq M$ there exist $C,\delta>0$ such that 
    \begin{equation*}
        \abs{\partial_1^\alpha g(\xbf)K(\xbf,\xbf')}\leq C\e^{-\delta\abs{\xbf'}}
    \end{equation*}
    for all $\xbf,\xbf'\in E$.
\end{lemma}
\begin{proof}
    By Lemma~\ref{lem:K1} we may assume that $K_1$ is of the form \eqref{eq:kernels1}. 
    
    Hence, suppose that $j_0\geq 2$. If $j_0>2$ then by applying Lemma~\ref{lem:K2} (with $N$ replaced with $j_0-1$) we may write
    \begin{align*}
        &\int_E\dots\int_E K_1(\xbf,\ybf_1)\dots K_{j_0-1}(\ybf_{j_0-2},\ybf_{j_0-1}) \dd \ybf_1\dots \dd\ybf_{j_0-2} \\
        &\qquad=H(\xbf,\ybf_{j_0-1})+\e^{\im b\phi(\xbf,\ybf_{j_0-1})}\sum_{n,\eta} a_{n,\eta}(\xbf) \mathfrak{f}_{n,\eta} (\xbf-\ybf_{j_0-1}),
    \end{align*}
    where $H$, $n$, $\eta$, and $a_{n,\eta}$ are as in Lemma~\ref{lem:K2}. By the properties of $H$ it is clear that we only need to consider the terms in the sum above. Hence it suffices to consider the case $j_0=2$, i.e.\ when $K_2$ is of the form \eqref{eq:kernels2}.

    By defining 
    \begin{equation*}
        \tilde{K}(\xbf,\ybf)=\int_{E} \e^{\im b\phi_s(\xbf-\ybf,\ybf-\xbf')}a_1(\xbf)\mathfrak{f}_{n_1,\alpha_1}(\xbf-\ybf)a_2(\ybf)h_2(\xbf-\ybf)\mathfrak{f}_{n_2,\alpha_2}(\ybf-(\xbf')^*) \dd \ybf
    \end{equation*}
    we may write 
    \begin{equation*}
        K(\xbf,\ybf)=\int_E\!\dots\!\int_E\!\! \e^{\im b\phi(\xbf,\ybf_2)}\tilde{K}(\xbf,\ybf_2)K_3(\ybf_2,\ybf_3)\dots K_N(\ybf_{N-1},\xbf')\dd \ybf_2\!\dots\! \dd  \ybf_{N-1}.
    \end{equation*}
    Hence, by \eqref{eq:qintegral} and Lebesgue's dominated convergence theorem, the proof is complete if we can show that 
    for all $c>0$ and $\alpha\in \N_0$ exist $C,\delta>0$ such that 
    \begin{equation}\label{eq:F1}
        \abs{\partial_1^\alpha \tilde{K}(\xbf,\xbf')}\leq C\e^{-\delta\abs{\xbf-(\xbf')^*}}
    \end{equation}
    for all $\xbf\in E$ with $x_2>c$ and $\xbf'\in E$.

    Then, since  the derivatives of $\mathfrak{f}_{n_1,\alpha_1}$ are only well-behaved away from the origin (cf. \eqref{eq:ffrak}), we split the domain of integration in two parts, one near $\xbf$ and one away from $\xbf$. After that, we treat the the first part by shifting the $\xbf$ dependency from $K_1$ to $K_2$ with a change of variable in order to exploit the presence of $\ybf^*$ in $K_2(\xbf,\ybf)$. 
    
    Let $c>0$ be arbitrary and let $\chi\in \Ccinf{\R^2}$, $\chi(\xbf)\in [0,1]$, be a function which equals $1$ on $B_{c/4}(0)$ and is supported in $B_{c/2}(0)$. Using $\chi$ to split the integral in two parts and applying a change of variable gives
    \begin{align*}
        \tilde{K}(\xbf,\xbf')&=a_1(\xbf)\int_{\R^2} \bigg[\e^{\im b\phi_s(\wbf,\xbf-\xbf'-\wbf)}\chi(\wbf)\mathfrak{f}_{n_1,\alpha_1}(\wbf)a_2(\xbf-\wbf)\\
        &\phantom{=a_1(\xbf)\int_{\R^2} \bigg[} \cdot h_2(\xbf-\wbf-\xbf')\mathfrak{f}_{n_2,\alpha_2}(\xbf-\wbf-(\xbf'^*))\bigg]\dd \wbf\\
        &\quad + a_1(\xbf)\int_{E} \bigg[\e^{\im b\phi_s(\xbf-\ybf,\ybf-\xbf')}(1-\chi(\xbf-\ybf))\mathfrak{f}_{n_1,\alpha_1}(\xbf-\ybf)a_2(\ybf)\\
        &\phantom{\quad + a_1(\xbf)\int_{E} \bigg[} \cdot h_2(\ybf-\xbf')\mathfrak{f}_{n_2,\alpha_2}(\ybf-(\xbf')^*) \bigg]\dd \ybf\\
        &\eqqcolon G_1(\xbf,\xbf')+G_2(\xbf,\xbf').
    \end{align*}
    We show first that \eqref{eq:F1} holds when $\tilde{K}$ is replaced by $G_1$. For any $\xbf\in E$ with $x_2>c$ and any $\wbf\in \Supp \chi$ it is clear that $x_2-w_2>c/2$ and hence $
        \abs{\xbf-\wbf-\xbf'}\leq\abs{\xbf-\wbf-(\xbf')^*} $
    for all $\xbf'\in E$. Thus, we get
    \begin{align*}
        \abs{\partial^\alpha_1 \e^{\im b\phi_s(\wbf,\xbf-\wbf-\xbf')}a_2(\xbf-\wbf)h_2(\xbf-\wbf-\xbf')}&\leq C\jnorm{\wbf}^k\jnorm{\xbf-\wbf-(\xbf')^*}^k 
    \end{align*}
    for some constants $C,k>0$. Hence by the integrability of $\mathfrak{f}_{n_1,\alpha_1}$, the exponential decay in $\abs{\xbf-\wbf-(\xbf')^*}$ of all derivatives of $\mathfrak{f}_{n_2,\alpha_2}(\xbf-\wbf-(\xbf')^*)$ (cf.\ \eqref{eq:ffrak}), and Lebesgue's dominated convergence theorem, it follows that $G_1$ satisfies \eqref{eq:F1}.

    Similar reasoning gives that also $G_2$ satisfies \eqref{eq:F1}, thus the proof is complete.
\end{proof}

\vfill

{\footnotesize
	
	\begin{tabular}{rl}
		(M. Moscolari) & \textsc{Fachbereich Mathematik, Eberhard-Karls-Universit\"at}\\
		& Auf der Morgenstelle 10, 72076 T\"ubingen, Germany \\
		&  \textsl{E-mail address}: \href{mailto:massimo.moscolari@mnf.uni-tuebingen.de}{\texttt{massimo.moscolari@mnf.uni-tuebingen.de}} \\
		\\
		(B.B. St{\o}ttrup) & \textsc{Department of Mathematical Sciences, Aalborg University} \\
		&  Skjernvej 4A, 9220 Aalborg, Denmark \\
		&  \textsl{E-mail address}: \href{mailto:benjamin@math.aau.dk}{\texttt{benjamin@math.aau.dk}} \\
	\end{tabular}
	
}

\end{document}